\documentclass[11p,reqno]{amsart}

\usepackage{graphicx}
\usepackage{amssymb,amsthm,amsmath,mathrsfs,bm,braket,marginnote}
\usepackage{enumerate}
\usepackage{appendix}
\usepackage[colorlinks=true, pdfstartview=FitV, linkcolor=blue, citecolor=blue, urlcolor=blue]{hyperref}
\usepackage{multirow}

\usepackage{pgf}
\usepackage{pgfplots}
\usepackage{tikz}
\usetikzlibrary{arrows,calc}
\usepackage{verbatim}
\usetikzlibrary{decorations.pathreplacing,decorations.pathmorphing}
\usepackage[numbers,sort&compress]{natbib}
\usepackage{dsfont}

\numberwithin{equation}{section}
\newtheorem{theorem}{Theorem}[section]
\newtheorem{lemma}[theorem]{Lemma}
\newtheorem{definition}[theorem]{Definition}
\newtheorem{remark}[theorem]{Remark}
\newtheorem{coro}[theorem]{Corollary}

\newtheorem{proposition}[theorem]{Proposition}

 \reversemarginpar
\DeclareMathOperator*{\Pf}{Pf}

\newcommand{\pt}{{\partial_t}}

\begin{document}

\title[Quasi-Pfaffians and applications]{Quasi-Pfaffians and applications}

\date{}

\dedicatory{}


\author{Claire Gilson}
\address{School of Mathematics and Statistics, University of Glasgow, Glasgow G12 8SQ, UK}
\email{claire.gilson@glasgow.ac.uk}

\author{Shi-Hao Li}
\address{Department of Mathematics, Sichuan University, Chengdu, 610064, PR China}
\email{shihao.li@scu.edu.cn}

\author{Guo-Fu Yu}
\address{School of Mathematical Sciences, Shanghai Jiaotong University, People's Republic of China; School of Mathematical Sciences, CMA-Shanghai, Shanghai Jiao Tong University, Shanghai 200240}
\email{gfyu@sjtu.edu.cn}

\begin{abstract}
\noindent
This paper presents a non-commutative generalization of the Pfaffian which we call a quasi-Pfaffian. 
This novel concept arises from solving linear systems with non-commutative skew-symmetric coefficients. 
A new non-commutative integrable system whose solutions are expressed in terms of these quasi-Pfaffians is presented. Derivative formulae and identities satisfied by these quasi-Pfaffians are presented. 
 
\end{abstract}

\maketitle

\section{Introduction}

The Pfaffian, named after Johann Friedrich Pfaff,  is an algebraic object associated with skew-symmetric matrices. Its introduction is credited to Arthur Cayley, who discovered that the determinant of any $2n \times 2n$ skew-symmetric matrix $A$ could be written as the square of a specific polynomial in the matrix entries \cite{cayley49}. 
Formally, for a skew-symmetric matrix $A$  with entries  $a_{ij}$, $i,j=1,\cdots,{2n}$, this is captured by the fundamental identity $\Pf(A)^2=\det(A)$. 

While the Pfaffian can simply be defined as a square root of the determinant for skew-symmetric matrices, its conceptual scope extends beyond this relationship. Knuth underscores this point by asserting that ``Pfaffians are more fundamental than determinants, in the sense that determinants are merely the bipartite special case of a general sum over matchings'' \cite{knuth96}. This conceptual primacy explains why the Pfaffian itself often appears in  mathematical and physical contexts. 
The utility of the Pfaffian is exemplified in combinatorics and statistical physics. One of its foremost applications is for counting perfect matchings in graphs, a method formalized by the FKT algorithm, which is fundamental to dimer statistics and the study of domino tilings in rectangular grids \cite{kasteleyn61,kenyon05}. Furthermore, the Pfaffian provides powerful techniques in enumerative combinatorics, including the counting of shifted Young diagrams, the analysis of non-intersecting paths, and the theory of Schur Q-functions \cite{stembridge90,ishikawa96}.

In these applications, the Pfaffian is not treated as a square root of the determinant but is instead defined by its intrinsic combinatorial nature. 
 Specifically, it is defined as
\begin{align}\label{expansion}
\Pf(A):=\Pf(1,\cdots,2n)=\sum_{P} (-1)^P \Pf(i_1,i_2)\Pf(i_3,i_4)\cdots\Pf(i_{2n-1},i_{2n}),
\end{align}
where $\Pf(i,j)=a_{i,j}$ are the entries in the Pfaffian. 
The sum is over all permutations $P$ of the set $\{1,\cdots,2n\}$ into $n$ disjoint pairs satisfying the ordering condition 
\begin{align}
i_1<i_2, \: i_3<i_4, \,\cdots,i_{2n-1}<i_{2n} 
\end{align}
and $i_1<i_3<\cdots<i_{2n-1}$, and $(-1)^P$ denotes the sign of the permutation associated with this ordered pairing.

Like determinants, Pfaffians obey bilinear identities.  These identities are of fundamental importance in integrable systems.   First investigated by Tanner \cite{tanner78}, the simplest identity states that 
\begin{align}\label{identity1}
\begin{aligned}
\Pf(\bullet,a_1,a_2,a_3,a_4)\Pf(\bullet)&=\Pf(\bullet,a_1,a_2)\Pf(\bullet,a_3,a_4)\\
&-\Pf(\bullet,a_1,a_3)\Pf(\bullet,a_2,a_4)+\Pf(\bullet,a_1,a_4)\Pf(\bullet,a_2,a_3).
\end{aligned}
\end{align}
In the above notation $\bullet=\{1,\cdots,2n\}$ are an arbitrary set of $2n$ labels. The proof of this identity was later furnished by Zajaczkowski \cite{zajaczkowski80}.
The above identity was then generalized by Perk et al., who systematized these relations into the ``compound Pfaffian theorem'' \cite[Sec. 3.2]{perk84},
culminating in the most general formula given by 
\begin{align*}
&\sum_{j=0}^M (-1)^j \Pf(b_0,\cdots,\hat{b}_j,\cdots,b_M )\Pf(b_j,c_0,\cdots,c_N )\\
&\quad=\sum_{k=0}^N (-1)^k \Pf(b_0,\cdots, b_M,c_k )\Pf(c_0,\cdots,\hat{c}_k,\cdots,c_N ).
\end{align*} 
In the above equation, $b_0,\cdots,b_M$ and $c_0,\cdots,c_N$ are $M+1$ and $N+1$ different labels, where $M,\,N\in2\mathbb{N}$. 
The general Pfaffian bilinear identity was proved by Ohta in his PhD thesis \cite{ohta92}, and later proved by Dress and Wenzel making the use of exterior algebra and $\Delta$-matroids \cite{dress95}. 
Pfaffians bilinear identities have been instrumental in the theory of integrable systems \cite{hirota04,ohta04,hirota00,tsujimoto96},  in particular  the $\tau$-functions for the BKP and DKP hierarchies are naturally expressed as Pfaffians.

Over the years, with the development of non-commutative geometry and quantum algebra, Pfaffians with non-commutative entries have been investigated. Ishikawa and Wakayama considered  minor summation formulae for Pfaffians in the framework of the quantum matrix algebra $\mathcal{A}(\text{Mat}_q(m,n))$ \cite{ishikawa95}. In \cite{jing14}, these summation formulae were called quantum Pfaffians, and explicit computations using quantum Pfaffians were carried out by using recursion formulae.

 In this paper, we consider a novel kind of non-commutative generalization which we call a quasi-Pfaffian. The idea was inspired by the work of Gelfand, Retakh et al, where quasi-determinants were introduced  \cite{gelfand91,gelfand92,gelfand05}. The quasi-Pfaffian turns out to be the most natural object for playing the role of non-commutative Pfaffian.   In what follows, we present the main results of this work.

\subsection{Main results}
 
We give a formal definition of the quasi-Pfaffian, with elements defined in a division ring $\mathcal{R}$ endowed with an anti-involution in Section \ref{sub2.2}. We demonstrate that the quasi-Pfaffian can be written as a ratio of Pfaffians in terms of commutative entries. It is crucial to note that, unlike the classical Pfaffian, the quasi-Pfaffian is not constrained to be zero when the last two labels are identical. 
In Section \ref{sub2.3}, we generalize a minor expansion formula for quasi-determinants to derive a corresponding  identity for quasi-Pfaffians \cite{krob95}. This result serves as a non-commutative generalization of identity \eqref{identity1}. A central finding, presented in Theorem \ref{thm2.3}, states that a quasi-Pfaffian of quasi-Pfaffians is itself a quasi-Pfaffian. This property demonstrates that quasi-Pfaffians adhere to a ``heredity principle''.  Furthermore, from a computational perspective, this theorem provides a condensation algorithm for quasi-Pfaffians, generalizing the known method for classical Pfaffians \cite{li20}.

%

In Section \ref{sec3}, we introduce the definitions of Grammian and Wronskian type quasi-Pfaffians and derive their key derivative formulae.
We show that the derivatives of these quasi-Pfaffians remain expressible as quasi-Pfaffians, this  enables the construction of a closed-form for ``non-commutative $\tau$-functions". 
Building upon these derivative formulas and the quasi-Pfaffian identity, in Section \ref{sec4}  the construction of a non-commutative B-Toda lattice whose solutions are explicitly given by quasi-Pfaffians is presented. Furthermore, we establish the Lax integrability of this system by constructing a family of $\mathcal{R}$-valued polynomials that satisfy a 6-term recurrence relation.

\section{Quasi-determinants and quasi-Pfaffians}\label{sec2}
This section presents the formal definition of quasi-Pfaffians. To motivate this definition, we consider in the appendix the solution to a linear system $A \underline{x}=\underline{b}$, where $A$ is a skew-symmetric matrix with non-commutative entries.
\subsection{Quasi-determinants}
We start with quasi-determinants for later use. The usual  definition of a quasi-determinant is
\begin{equation}\label{exform1}
|A|_{ij} = a_{ij} - r^j_i \left( A^{ij} \right)^{-1} c^j_i, 
\end{equation}
where $ r^j_i $ is the $i$-th row of $ A$ with the  $ j $-th element removed, \( c^j_i \) is the \( j \)-th column with the \( i \)-th element removed, and $A^{ij}$  is the matrix obtained by removing both the $ i$-th row and the  $j$-th column from the matrix $ A $, assuming that $A^{ij}$ is invertible in a certain free generated ring.  This can also be expressed simply in terms of a boxed element by
\[
|A|_{ij} =
\left|
\begin{array}{cc}
A^{ij} & c^j_i \\
r^j_i & \boxed{a_{ij}}
\end{array}
\right|.
\]
 For a detailed introduction of quasi-determinants, please refer to Gelfand et al   \cite{gelfand05}.

In the case where entries in the quasi-determinant commute we have the formula
$|A|_{ij}=(-1)^{i+j}\det(A)/\det(A^{ij})$. 
For a block matrix 
\begin{equation*}\label{block}
 \begin{pmatrix}
      A&B\\
      C&d
    \end{pmatrix}
\end{equation*}
where $d$ is a single entry, $A$ is a square matrix
 (say $n \times n$) and $B$, $C$ are column and row vectors,
the quasi-determinant expanding about the last entry is defined as 
\begin{equation*}\label{block}
 \begin{vmatrix}
      A&B\\
      C&\fbox{$d$}
    \end{vmatrix}
    =d-CA^{-1} B.
\end{equation*}
This object is a very natural one as the last entry in the inverse of the block matrix is given by the inverse of the quasi-determinant
\begin{equation*}
\left\{\begin{pmatrix}
      A&B\\
      C&d
    \end{pmatrix}^{-1}\right\}_{n+1,n+1}=(d-CA^{-1} B)^{-1}.
\end{equation*}


An important result in the theory of quasi-determinants is  Sylvester's identity \cite[Thm 1.5.2]{gelfand05}.  This is the non-commutative form of Jacobi's identity for determinants.
\begin{align*}
\left|\begin{array}{ccc}
A&c_1&c_2\\
r_1& d_{11}& d_{12}\\
r_2&d_{21}& \boxed{d_{21}}\end{array}
\right|
=
\left|\begin{array}{cc}
\left|\begin{array}{cc}
A&c_1\\
r_1&\boxed{d_{11}}\end{array}
\right|
&\left|\begin{array}{cc}
A&c_2\\
r_1&\boxed{d_{12}}\end{array}
\right|\\[1em]
\left|\begin{array}{cc}
A&c_1\\
r_2&\boxed{d_{21}}\end{array}
\right|
&
\boxed{\left|\begin{array}{cc}
A&c_2\\
r_2&\boxed{d_{22}}\end{array}
\right|}
\end{array}\right|,
\end{align*}
where $A$ is a square matrix, $c_i$ and $r_i$ are single column  and row vectors respectively and $d_{ij}$ are single entries.

\subsection{Definition of quasi-Pfaffians}\label{sub2.2}
Assume that $\mathcal{R}$ is a division ring equipped with an anti-involution $\top$ such that for any $a,b\in \mathcal{R}$,  they satisfy $(a+b)^\top=a^\top+b^\top$, $(ab)^\top=b^\top a^\top$ and $(a^\top)^\top=a$. 
The skew-symmetry of a matrix over $\mathcal{R}$ is then defined under anti-involution. 
We define the anti-involution of $A=(a_{ij})_{1\leq i,j\leq n}$ as
\begin{align*}
\left(\begin{array}{ccc}
a_{11}&\cdots&a_{1n}\\
\vdots&&\vdots\\
a_{n1}&\cdots&a_{nn}\end{array}
\right)^\top=\left(\begin{array}{ccc}
a_{11}^\top&\cdots&a_{n1}^\top\\
\vdots&&\vdots\\
a_{1n}^\top&\cdots&a_{nn}^\top\end{array}
\right).
\end{align*}
A matrix $A$ is called skew symmetric if $A^\top=-A$. 

A quasi-Pfaffian is a quasi-determinant where the main body of the matrix is a skew-symmetric matrix. We give a precise definition below.
\begin{definition} 
Given a skew-symmetric matrix  $A=(a_{i,j})_{i,j=1}^{2n}\in\mathcal{R}^{2n\times 2n}$ which is invertible in $\mathcal{R}$. We can define a quasi-Pfaffian as
\begin{align}\label{def}
\Pf(1,&\cdots,2n ,\boxed{2n+1,2n+2})=
\left|
\begin{array}{cccc}
a_{1,1}&\cdots&a_{1,2n}&a_{1,2n+2}\\
\vdots&&\vdots&\vdots\\
a_{2n,1}&\cdots&a_{2n,2n}&a_{2n,2n+2}\\
a_{2n+1,1}&\cdots&a_{2n+1,2n}&\boxed{a_{2n+1,2n+2}}
\end{array}
\right|.
\end{align}
\end{definition}
\begin{remark}
The numbers $\{1,\cdots,2n+2\}$are purely a set of labels. One can take the labels as arbitrary numbers or letters in an alphabet, such as $d_0,d_1,\cdots$. In the subsequent sections, the symbol $\bullet$ is reserved as a shorthand for a set of an even number of indices. This notation allows us to express Pfaffians more compactly; for instance, we can write $\Pf(\bullet, \boxed{
2n+1,2n+2})$ to represent $\Pf(1,\dots,2n, \boxed{2n+1, 2n+2})$ without any ambiguity.
\end{remark}

In the above definition of the quasi-Pfaffian we don't necessarily assume that every element in the quasi-Pfaffian is invertible.
However, we should require that the principle minors of even size are invertible, i.e. $(a_{ij})_{i,j=1}^{2k}$, $k=1,2,\cdots,n$ are invertible. 
Using  the standard quasi-determinant expansion formula \eqref{exform1}  we have
\begin{align}\label{expand}
\Pf(1,\cdots,2n,\boxed{2n+1,2n+2})=a_{2n+1,2n+2}-  \sum_{i,j=1}^{2n }    a_{2n+1,i}  (A^{-1})_{ij}a_{j,2n+2}.
\end{align}
Therefore, to explicitly get the value of quasi-Pfaffian, we need to compute the inverse of $A$, which is a 2-step calculation. The first step is to write down the inverse of $A$ in terms of inverse of some 2-by-2 matrix. Using the block matrix representation 
\begin{align*}
A=\left(\begin{array}{cc}
A_{11}&B\\
-B^\top&A_{22}
\end{array}
\right),
\end{align*} where $A_{11}$ is a 2-by-2 matrix, and $A_{22}$ is a $(2n-2)$-by-$(2n-2)$ invertible matrix. With principle minors being invertible, we could have 
\begin{align*}
\left(\begin{array}{cc}
A_{11}&B\\
-B^\top&A_{22}
\end{array}
\right)^{-1}=\left(\begin{array}{cc}
(A_{11}+BA_{22}^{-1}B^\top)^{-1}&-(A_{11}+BA_{22}^{-1}B^\top)^{-1}BA_{22}^{-1}\\
A_{22}^{-1}B^\top(A_{11}+BA_{22}^{-1}B^\top)^{-1}&A_{22}^{-1}-A_{22}^{-1}B^\top(A_{11}+BA_{22}^{-1}B^\top)^{-1}BA_{22}^{-1}
\end{array}
\right).
\end{align*}
Following the procedure, we can decompose the skew symmetric $(2n-2)$-by-$(2n-2)$ matrix into lower order matrices, until all matrices in $A^{-1}$ are 2-by-2 matrices. Then we go to the second step, where we need to compute the inverse of 2-by-2 skew symmetric matrices with non-commutative entries.  
We compute the inverse by considering an expansion about non-diagonal entries
\begin{align}\label{inverse1}
\left(\begin{array}{cc}
a_{11}&a_{12}\\
a_{21} &a_{22}
\end{array}
\right)^{-1}
=\left(\begin{array}{cc}
-a_{21}^{-1}a_{22}Q^{-1}     &   a_{21}^{-1}(1+  a_{22}Q^{-1}a_{11} a_{21}^{-1})\\
Q^{-1}                                       &                              -Q^{-1} a_{11}  a_{21}^{-1}
\end{array}
\right),
\end{align}
where $Q=a_{12}-a_{11}a_{21}^{-1}  a_{22}$, the quasi-determinant expanding about $a_{12}$.
If our matrix has  $a_{ji}=-a_{ij}^\top$, i.e. the structure of a skew-symmetric matrix with non-commutative entries and additionally taking the diagonal elements to be zero, i.e.  $a_{11}=a_{22}=0$  then
\begin{align*}
\left(\begin{array}{cc}
0&a_{12}\\
-a_{12}^\top &0
\end{array}
\right)^{-1}=\left(\begin{array}{cc}
0&(-a_{12}^{\top})^{ -1}\\
a_{12}^{-1} &0
\end{array}
\right).
\end{align*}
As a simple example, we have
\begin{align*}
&\Pf(1,2,\boxed{3,4})=
a_{34}-\left(
a_{31}, a_{32}
\right)\left(\begin{array}{cc}
a_{11}&a_{12}\\
a_{21}&a_{22}
\end{array}
\right)^{-1}
\left(\begin{array}{c}
a_{14}\\a_{24}
\end{array}\right)\\
&\quad=
a_{34}-\left(
a_{31}, a_{32}
\right)
\left(\begin{array}{cc}
-a_{21}^{-1}a_{22}Q^{-1}     &   a_{21}^{-1}(1+  a_{22}Q^{-1}a_{11} a_{21}^{-1})\\
Q^{-1}                                       &                              -Q^{-1} a_{11}  a_{21}^{-1}
\end{array}
\right)
\left(\begin{array}{c}
a_{14}\\a_{24}
\end{array}\right)\\
&\quad=
a_{34}-\left(a_{31}a_{12}^{-\top}a_{22}+a_{32}\right)Q^{-1}a_{14}+\left(a_{31}a_{12}^{-\top}(-a_{12}^\top+a_{22}Q^{-1}a_{11})^{-1}a_{12}^\top+a_{32}Q^{-1}a_{11}a_{12}^{-\top}\right)a_{24},
\end{align*}
where $a_{ij}^{\top}=-a_{ji}$ and $Q=a_{12}+a_{11}a^{-\top}_{12}a_{22}$.
In the commutative case, with $a_{11}=a_{22}=0$ and $Q=a_{12}$, the above formula simplifys to
\begin{align*}\Pf(1,2,\boxed{3,4})=a_{34}+a_{31}a^{-1}_{12}a_{24}-a_{32}a_{12}^{-1}a_{14}=\frac{\Pf(1,2,3,4)}{\Pf(1,2)}.
\end{align*}

\begin{proposition}
In the commutative setting,
quasi-Pfaffian defined by \eqref{def} is a ratio of Pfaffians. In other words, we have
\begin{align}\label{nc-c}
\Pf(1,\cdots,2n,\boxed{2n+1,2n+2})=\frac{\Pf(1,\cdots,2n,2n+1,2n+2)}{\Pf(1,\cdots,2n)}.
\end{align}
\end{proposition}
\begin{proof}
This proposition could be verified by making the use of Cayley's formula \eqref{cayley}. The formulae $
\det(A)=\Pf(1,\cdots,2n)^2$, and 
\begin{align*}
\det\left(\begin{array}{cccc}
a_{1,1}&\cdots&a_{1,2n}&a_{1,2n+2}\\
\vdots&&\vdots&\vdots\\
a_{2n,1}&\cdots&a_{2n,2n}&a_{2n,2n+2}\\
a_{2n+1,1}&\cdots&a_{2n+1,2n}&a_{2n+1,2n+2}
\end{array}
\right)=\Pf(1,\cdots,2n)\Pf(1,\cdots,2n+2)
\end{align*}
leads to the result.
\end{proof}

In the commutative case, we know that the determinant of a skew symmetric matrix of odd size is zero. Thus if our final row and column in our commutative quasi-pfaffian are the same with each other, then the  Pfaffian in the numerator of  \eqref{nc-c}  will be zero. 
However, for quasi-Pfaffians, we have the following proposition, showing its difference with the commutative case. 
\begin{proposition}\label{lem1}
Let $\bullet$ be a set of $2m$ indices, and let $i, j$ be any two indices. Then,
\begin{align}
\Pf(\bullet, \boxed{
i,j}) = -\Pf(\bullet, \boxed{j,i})^\top.
\end{align}
\end{proposition}
\begin{proof}
Let us denote the skew symmetric matrix related to the index set as $A$, then according to the definition, we have
\begin{align*}
\Pf(\bullet,\boxed{i,j})=a_{i,j}-\sum_{k,l} a_{i,k}(A^{-1})_{k,l}a_{l,j},
\end{align*}
then by noting that $\top$ is an anti-involution over $\mathcal{R}$, from which we get
\begin{align*}
\Pf(\bullet,\boxed{i,j})^\top=-a_{j,i}-\sum_{k,l} a_{j,l}(A^{-\top})_{k,l}a_{k,i}=-a_{j,i}+\sum_{k,l}a_{j,l}(A^{-1})_{l,k}a_{k,i}=-\Pf(\bullet,\boxed{j,i}).
\end{align*}
\end{proof}
\begin{coro}
The quasi-Pfaffian $\tau = \Pf(\bullet, \boxed{
i,i})$ is antisymmetric under anti-involution $\tau^\top = -\tau$, but is not necessarily zero.
\end{coro}
A key feature of skew-symmetric matrices with non-commutative entries is that their quasi-Pfaffians need not vanish. This is exemplified by a 3-by-3 skew-symmetric matrix over the real quaternions, which, as shown in \cite[eq (4.9)]{lam06}, can be full-rank. Consequently, its last row is linearly independent of the first two, implying that its quasi-Pfaffian is non-zero.

Moreover, we have the following proposition, which states that a quasi-Pfaffian vanishes if the indices in its final block exhibit dependence with those in the preceding.
\begin{proposition}
If the index $i$ coincide with an index in $\bullet$, then for any $j$, we have
\begin{align}\label{zerocondition}
\Pf(\bullet,\boxed{i,j})=\Pf(\bullet,\boxed{j,i})=0.
\end{align}
\end{proposition}
This equation is based on the basic property of quasi-determinants, where two rows/columns of a quasi-determinant coincide with each other except the ones with the boxed element. Besides, based on \cite[Thm 1.4.6]{gelfand05}, it is known that if the last row/column of the matrix is a left/right linear combination of the rest rows/columns, then its corresponding quasi-Pfaffian is also zero.


\subsection{A  quasi-Pfaffian identity}\label{sub2.3}
This part is devoted to an identity of quasi-Pfaffians, which could be viewed as a non-commutative generalization of the bilinear identity \eqref{identity1}. 
\begin{theorem}\label{thm2.3}
Quasi-Pfaffians satisfy the following identity 
\begin{align}\label{identity}
\begin{aligned}
\Pf(\bullet&,a,b,\boxed{c,d})
=\left|
\begin{array}{ccc}
\Pf(\bullet,\boxed{a,a})&\Pf(\bullet,\boxed{a,b})&\Pf(\bullet,\boxed{a,d})\\[1em]
\Pf(\bullet,\boxed{b,a})&\Pf(\bullet,\boxed{b,b})&\Pf(\bullet,\boxed{b,d})\\[1em]
\Pf(\bullet,\boxed{c,a})&\Pf(\bullet,\boxed{c,b})&\boxed{\Pf(\bullet,\boxed{c,d})}
\end{array}
\right|\\
&=\Pf(\bullet,\boxed{c,d})-\left(\Pf(\bullet,\boxed{c,a})\,\,\Pf(\bullet,\boxed{c,b})\right)\left(
\begin{array}{cc}
\Pf(\bullet,\boxed{a,a})&\Pf(\bullet,\boxed{a,b})\\[1em]
\Pf(\bullet,\boxed{b,a})&\Pf(\bullet,\boxed{b,b})
\end{array}
\right)^{-1}\left(\begin{array}{c}
\Pf(\bullet,\boxed{a,d})\\[1em]\Pf(\bullet,\boxed{b,d})
\end{array}
\right).
\end{aligned}
\end{align}
\end{theorem}
The proof of this theorem is based on a double application of  
Sylvester's theorem \cite[Section 2.3.4]{krob95}.
\begin{lemma}\label{lemmaa}
For an invertible square matrix $A$, rows $r_i$ and columns $c_i$ for $i=1,2,3$, and elements $d_{ij}$ for $i,j=1 \cdots ,3$, we have the following quasi-determinant  identity
\begin{align}\label{qd id3x3}
\left|\begin{array}{cccc}
A&c_1&c_2&c_3\\
r_1&d_{11}&d_{12}&d_{13}\\
r_2&d_{21}&d_{22}&d_{23}\\
r_3&d_{31}&d_{32}&\boxed{d_{33}}
\end{array}
\right|=\left|\begin{array}{ccc}
\left|\begin{array}{cc}
A&c_1\\r_1&\boxed{d_{11}}
\end{array}
\right|&\left|\begin{array}{cc}
A&c_2\\r_1&\boxed{d_{12}}
\end{array}
\right|&\left|\begin{array}{cc}
A&c_3\\
r_1&\boxed{d_{13}}
\end{array}
\right|\\[1em]
\left|\begin{array}{cc}
A&c_2\\r_1&\boxed{d_{21}}
\end{array}
\right|&\left|\begin{array}{cc}
A&c_2\\r_2&\boxed{d_{22}}
\end{array}
\right|&\left|\begin{array}{cc}
A&c_3\\r_2&\boxed{d_{23}}
\end{array}
\right|\\[1em]
\left|\begin{array}{cc}
A&c_1\\r_3&\boxed{d_{31}}
\end{array}
\right|&
\left|\begin{array}{cc}
A&c_2\\r_3&\boxed{d_{32}}
\end{array}
\right|&\boxed{\left|\begin{array}{cc}
A&c_3\\r_3&\boxed{d_{33}}
\end{array}
\right|}
\end{array}
\right|.
\end{align}
\end{lemma}
\begin{proof}
Using Sylvester's identity, expanding out the left hand side of \eqref{qd id3x3} we get
 \begin{align}\label{qd id2x2}
\left|
   \begin{array}{cccc}
      A&c_1&c_2&c_3\\
      r_1&d_{11}&d_{12}&d_{13}\\
      r_2&d_{21}&d_{22}&d_{23}\\
      r_3&d_{31}&d_{32}&\boxed{d_{33}}
   \end{array}
\right|
=
\left|
\begin{array}{cc}
             \left|
                \begin{array}{ccc}
                     A&c_1&c_2\\r_1&d_{11}&d_{12}\\r_2&\boxed{d_{21}}& d_{22}
                \end{array}
            \right|
& 
         \left|
                \begin{array}{ccc}
                      A&c_2&c_3\\r_1&d_{1  2}& d_{13}\\r_2&d_{22}&\boxed{d_{23}}
                \end{array}
            \right|
\\[2em]
 \left|
                \begin{array}{ccc}
                      A&c_1&c_2\\r_1&d_{11}&d_{12}\\r_3&\boxed{d_{31}}&d_{32}
                \end{array}
            \right|
& 
       \boxed{  \left|
                \begin{array}{ccc}
                       A&c_2&c_3\\r_1&d_{12}&d_{13}\\r_3&d_{32}&\boxed{d_{33}}
                \end{array}
            \right|}
\end{array}
\right|.
\end{align}
Now  using Sylvester's identity again on the four quasideterminants on the righthand side of \eqref{qd id2x2} we will get the same as expression as applying Sylvester's identity to the right hand side of equation \eqref{qd id3x3}.
\end{proof}

To tie Lemma \ref{lemmaa} up with Theorem \ref{thm2.3} we take $\bullet$ as the index set related to $A$, which is a skew-symmetric matrix of even order, i.e. $A=- A^\top$  and  the columns $c_i$,  rows $r_i$ and single entries $d_{ij}$  are related  by $ r_i=- c_i^\top, d_{ij}=-d_{ji}^\top$ for $i,j=1,2$.  We relabel $c_3$ as $c_4$ , and  $d_{i3}$  as  $d_{i 4}$ for $i=1,2,3$.
Then this now gives us the quasi-Pfaffian identity easily. For a skew-symmetric matrix $A = (a_{i,j})_{1 \leq i,j \leq 2n}$, one can compute the quasi-Pfaffian recursively using the formula \eqref{identity}. This reflects the ``heredity principle'' for quasi-Pfaffians---a quasi-Pfaffian could be viewed as a quasi-Pfaffian of quasi-Pfaffians in lower orders. In fact, Theorem \ref{thm2.3} indicates a condensation algorithm for quasi-Pfaffians.

At the end of this section, we remark that in the commutative case by using \eqref{nc-c}, the above equation \eqref{identity} can be reduced to the commutative bilinear identity
\begin{align*}
\Pf(\bullet,a,b,c,d)\Pf(\bullet)=\Pf(\bullet,a,b)\Pf(\bullet,c,d)-\Pf(\bullet,a,c)\Pf(\bullet,b,d)+\Pf(\bullet,a,d)\Pf(\bullet,b,c).
\end{align*}

\section{Derivative formulae for quasi-Pfaffians}\label{sec3}

In this section, we demonstrate the derivative formulae for quasi-Pfaffians. Derivative formulae for standard Pfaffians are widely used in soliton theory \cite{hirota04}.  In the literature Pfaffian solutions to integrable systems are classified into two distinct categories: Grammian type Pfaffians and Wronskian type Pfaffians.

A Grammian type Pfaffian admits elements in the form  \cite[Equation (3.99)]{hirota04}
\begin{align*}
\Pf(i,j)=c_{ij}+\int^t (f_ig_j-f_jg_i)dx,\quad c_{ij}=-c_{ji}.
\end{align*}
Derivatives for the Grammian type elements can be split into two parts. If we simply denote $\Pf(c,i)=f_i$ and $\Pf(d,i)=g_i$, then we have the derivative formula
\begin{align}\label{gram}
\frac{d}{dt} \Pf(i,j)=\Pf(c,i)\Pf(d,j)-\Pf(c,j)\Pf(d,i)=\Pf(d,c,i,j),\quad \Pf(d,c)=0.
\end{align}
Formally, we call the Pfaffians whose elements satisfy \eqref{gram} as Grammian type Pfaffians. These Pfaffians arise in the description of solutions to the BKP hierarchy \cite{hirota89,tsujimoto96}. 

A Wronskian type Pfaffian is a Pfaffian whose entries $\Pf(i,j)$ satisfy differential rules with respect to the time flows $t_1,t_2,\cdots$ such that
\begin{align}\label{wronskian}
\frac{\partial}{\partial t_n}\Pf(i,j)=\Pf(i+n,j)+\Pf(i,j+n). 
\end{align}
It has been demonstrated \cite[Equation 3.84]{hirota04} that the derivative of such a Pfaffian mirrors the derivative of a Wronskian determinant. Consequently, we define any Pfaffian obeying \eqref{wronskian} to be of Wronskian type. This kind of Pfaffian is used to express the solutions of DKP and Pfaff lattice hierarchies \cite{hirota91,adler99,kodama10}.

In this section, we assume elements $a_{i,j}$ in quasi-Pfaffians have both Grammian and Wronskian properties, in the form\footnote{Here we use the Hirota's bilinear operator $\mathcal{D}$ defined by $\mathcal{D}_t f\cdot g=\left.\frac{\partial}{\partial t}f(t+y)g(t-y)\right|_{y=0}$.}
\begin{align*}
\Pf(i,j):=a_{ij}=\int^t \mathcal{D}_t\phi_i^\top(t)\cdot \phi_j(t) dt=\int^t \phi_{i+1}^\top(t)\phi_j(t)-\phi_i^\top(t)\phi_{j+1}(t)dt,
\end{align*}
where $\phi_i: \mathbb{R}\to\mathcal{R}$ satisfying $\frac{d}{dt}\phi_i=\phi_{i+1}$. We may take this family of functions to be the moment functions, where
\begin{align*}
\phi_i(t):=\int_{\mathbb{R}}x^i W(x)e^{xt}d\mu(x),
\end{align*} 
and $W(x)$ is a $\mathcal{R}$-valued function equipped with the anti-involution $\top$. Indeed, we need to assume that $W(x)$ is symmetric under the involution, namely $W^\top(x)=W(x)$. 
The convergence of the integral is ensured by the choice of measure and an appropriate weight function  $W(x)$.
One can easily check that $a_{ij}=-a_{ji}^\top$, and that $a_{ij}$ satisfy both Wronskian and Grammian properties
\begin{align*}
\frac{d}{dt} \Pf(i,j)=\Pf(i+1,j)+\Pf(i,j+1),\quad \frac{d}{dt}\Pf(i,j)=\Pf(i,d_1)\Pf(j,d_0)^\top-\Pf(i,d_0)\Pf(j,d_1)^\top,
\end{align*}
where we introduce the notation $\Pf(i,d_j)=\phi_{i+j}^\top$ to represent moment functions. Besides, we can define $$\Pf(d_j,i)=-\Pf(i,d_j)^\top=-\phi_{i+j}.$$
Therefore, if we introduce the notation
\begin{align*}
\Pf(0,\cdots,2n-1,\boxed{2n,d_j})=\left|\begin{array}{cccc}
a_{0,0}&\cdots&a_{0,2n-1}&\phi_j^\top\\
\vdots&&\vdots&\vdots\\
a_{2n-1,0}&\cdots&a_{2n-1,2n-1}&\phi_{2n-1+j}^\top\\
a_{2n,0}&\cdots&a_{2n,2n-1}&\boxed{\phi_{2n+j}^\top}
\end{array}
\right|,
\end{align*}
then similar to Lemma \ref{lem1}, we have the following corollary.
\begin{coro}
Let $\bullet$ be a label set in $\mathbb{N}$ with even numbers. For any $i\in\mathbb{N}$, we have
\begin{align}\label{rel1}
\Pf(\bullet,\boxed{i,d_j})^\top=-\Pf(\bullet,\boxed{d_j,i}).
\end{align}
\end{coro}

\subsection{Derivative formulae for Wronskian type quasi-Pfaffians}
We start with the derivative formulae for Wronskian-type quasi-Pfaffians. Provided the elements satisfy
\begin{align}\label{wronskian}
\begin{aligned}
&\frac{d}{dt}\Pf(i,j)=\int^t \mathcal{D}_t \phi_{i+1}^\top\cdot\phi_jdt+\int^t \mathcal{D}_t\phi_i^\top\cdot\phi_{j+1}dt=\Pf(i+1,j)+\Pf(i,j+1),\\& \frac{d}{dt}\Pf(i,d_j)=\phi_{i+j+1}^\top=\Pf(i+1,d_j),
\end{aligned}
\end{align}
the derivative yields an expression in terms of quasi-Pfaffians themselves. 
\begin{theorem}\label{thm3.2}
For Wronskian-type Pfaffian satisfying \eqref{wronskian}, one has
\begin{align}\label{wronskian1}
\begin{aligned}
\frac{d}{dt}\Pf(\bullet,\boxed{i,j})&=\Pf(\bullet,\boxed{i+1,j})+\Pf(\bullet,\boxed{i,j+1})\\
&+\Pf(\bullet,\boxed{i,c_{2n-1}})\Pf(\bullet,\boxed{2n,j})-\Pf(\bullet,\boxed{i,2n})\Pf(\bullet,\boxed{c_{2n-1},j}),
\end{aligned}
\end{align}
where we denote $\bullet=\{0,1,\cdots,2n-1\}$ for short, and the label $c_i$ is given by $
\Pf(i,c_j)=\delta_{i,j},
$
and
\begin{align*}
\Pf(0,\cdots,2n-1,\boxed{i,c_{2n-1}})=\left|\begin{array}{cccc}
a_{0,0}&\cdots&a_{0,2n-1}&0\\
\vdots&&\vdots&\vdots\\
a_{2n-1,0}&\cdots&a_{2n-1,2n-1}&1\\
a_{i,0}&\cdots&a_{i,2n-1}&\boxed{\delta_{i,2n-1}}\end{array}
\right|.
\end{align*}
\end{theorem}
\begin{proof}
Firstly, we have the equation 
\begin{align*}
&\frac{d}{dt}\Pf(\bullet,\boxed{i,j})=
\frac{d}{dt}
\left(a_{i,j}-     a_i^{0\to 2n-1} {A^{-1}}\; a^j_{0\to 2n-1}\right)\\
&\quad=
a_{i+1,j}+a_{i,j+1}-a_i^{1\to 2n}A^{-1}a^j_{0\to 2n-1}-a_i^{0\to 2n-1}A^{-1}a^j_{1\to 2n}\\
&\quad\quad \,\,
-a_{i+1}^{0\to 2n-1}A^{-1}a^j_{0\to 2n-1}-a_i^{0\to 2n-1}A^{-1}a^{j+1}_{0\to 2n-1} + a_i^{0\to 2n-1}A^{-1}\left(\frac{d}{dt} A\right)A^{-1}a^j_{0\to 2n-1},
\end{align*}
where $a_i^{k\to l}$ (resp. $a_{k\to l}^i$) means a row (resp. column) vector from $a_{i,k}$ to $a_{i,l}$ (resp. $a_{k,i}$ to $a_{l,i}$), and $A=(a_{ij})_{i,j=0}^{2n-1}$.
Moreover, the above equation could be simplified as
\begin{align*}
\frac{d}{dt}\Pf(\bullet,\boxed{i,j})=
&\Pf(\bullet,\boxed{i+1,j})+\Pf(\bullet,\boxed{i,j+1})+\left|
\begin{array}{cc}
A&a_{0\to 2n-1}^{j}\\
a_{i}^{1\to 2n}&\boxed{0}
\end{array}
\right|+\left|
\begin{array}{cc}
A&a_{1\to 2n}^{j}\\
a_{i}^{0\to 2n-1}&\boxed{0}
\end{array}
\right|
\\
&\qquad\qquad+\sum_{k=1}^{2n}\left|
\begin{array}{cc}
A&e_k\\
a_{i}^{0\to 2n-1}&\boxed{0}
\end{array}
\right|\cdot\left|
\begin{array}{cc}
A&a_{0\to 2n-1}^{j}\\
a_{k}^{0\to 2n-1}&\boxed{0}
\end{array}
\right|
\\
&\qquad\qquad\qquad\qquad+\sum_{k=1}^{2n}\left|
\begin{array}{cc}
A&a_{0\to 2n-1}^{k}\\
a_{i}^{0\to 2n-1}&\boxed{0}
\end{array}
\right|\cdot\left|
\begin{array}{cc}
A&a_{0\to 2n-1}^{j}\\
e_k^\top&\boxed{0}
\end{array}
\right|,
\end{align*}
where $e_k$ is a unit column vector whose $k$-th element is $1$ (unity in $\mathcal{R}$), and others are zeros. $e_k^\top$ is the corresponding row vector. 
Using the property \eqref{zerocondition}, we have
\begin{align*}
\left|
\begin{array}{cc}
A&a_{0\to 2n-1}^{j}\\
a_{k}^{0\to 2n-1}&\boxed{0}
\end{array}
\right|=-a_{k,j},\quad 0\leq k\leq 2n-1,
\end{align*}
which results in the formula
\begin{align*}
&\sum_{k=1}^{2n}\left|\begin{array}{cc}
A&e_k\\a_i^{0\to2n-1}&\boxed{0}
\end{array}
\right|\cdot\left|\begin{array}{cc}
A&a_{0\to 2n-1}^j\\
a_k^{0\to 2n-1}&\boxed{0}
\end{array}
\right|+
\left|
\begin{array}{cc}
A&a_{1\to 2n}^j\\a_i^{0\to 2n-1}&\boxed{0}
\end{array}
\right|\\
&\quad=\left|
\begin{array}{cc}
A&e_{2n}\\
a_i^{0\to 2n-1}&\boxed{0}
\end{array}
\right|\cdot\left|
\begin{array}{cc}
A&a^j_{0\to 2n-1}\\
a_{2n}^{0\to 2n-1}&\boxed{a_{2n,j}}
\end{array}
\right|=\Pf(\bullet,\boxed{i,c_{2n-1}})\Pf(\bullet,\boxed{2n,j}).
\end{align*}
The desired result follows from the combination of above formulae.
\end{proof}

In the commutative case, this formula is equivalent to the equation 
\begin{align*}
\frac{d}{dt}\Pf(0,\cdots,2n-2,2n-1)=\Pf(0,\cdots,2n-2,2n),
\end{align*}
which is widely used in the construction of Wronskian-type Pfaffian solution for DKP hierarchy \cite[Sec. 3.4]{hirota04}, and the spectral transformation for skew-orthogonal polynomials \cite[Lemma 3.5]{adler99}. 
Similarly, consideration of the Pfaffian $\tau$-functions $\Pf(0,\cdots,2n-1,2n,d_0)$ yields the derivative formula \cite{chang18,li24}
\begin{align*}
\frac{d}{dt}\Pf(0,\cdots,2n-1,2n,d_0)=\Pf(0,\cdots,2n-1,2n+1,d_0).
\end{align*}
This formula is used to construct the B\"acklund transformation for DKP hierarchy. 
We now present the corresponding non-commutative analog.
\begin{theorem}
For arbitrary $i\in\mathbb{N}$, one has
\begin{align}\label{wronskian2}
\frac{d}{dt}\Pf(\bullet,\boxed{i,d_j})&=\Pf(\bullet,\boxed{i+1,d_j})-\Pf(\bullet,\boxed{i,2n})\Pf(\bullet,\boxed{c_{2n-1},d_j})+\Pf(\bullet,\boxed{i,c_{2n-1}})\Pf(\bullet,\boxed{2n,d_j}),
\end{align}
where $\bullet=\{0,\cdots,2n-1\}$ and $\Pf(d_i,c_j)=\Pf(c_j,d_i)=0$ for all $i,j\in\mathbb{N}$.
\end{theorem}
\begin{proof}
This proof is similar to that for Theorem \ref{thm3.2}. If we denote $\Phi_j^\top$ as a column vector composed of $\phi_j^\top,\phi_{j+1}^\top,\cdots,\phi_{j+2n-1}^\top$, then we have the derivative relation
\begin{align*}
\frac{d}{dt}\Pf(\bullet,\boxed{i,d_j})&=\phi_{i+j+1}^\top-\frac{d}{dt}(a_i^{0\to 2n-1}A^{-1}\Phi_j^\top)\\
&=\Pf(\bullet,\boxed{i+1,d_j})+\left|\begin{array}{cc}
A&\Phi_j^\top\\
a_i^{1\to2n}&\boxed{0}
\end{array}
\right|+\sum_{k=1}^{2n}\left|\begin{array}{cc}
A&e_k\\
a_i^{0\to 2n-1}&\boxed{0}
\end{array}
\right|\cdot\left|\begin{array}{cc}
A&\Phi_j^\top\\
a_k^{0\to 2n-1}&\boxed{0}
\end{array}
\right|\\
&+\sum_{k=1}^{2n}\left|\begin{array}{cc}
A&a_{0\to 2n-1}^k\\
a_i^{0\to 2n-1}&\boxed{0}
\end{array}
\right|\cdot\left|\begin{array}{cc}
A&\Phi_j^\top\\
e_k^\top&\boxed{0}\end{array}
\right|+\left|\begin{array}{cc}
A&\Phi_{j+1}^\top\\
a_i^{0\to 2n-1}&\boxed{0}
\end{array}
\right|.
\end{align*}
Moreover, according to the relation \eqref{zerocondition}, we have
\begin{align}\label{pfd0}
\left|\begin{array}{cc}
A&\Phi_j^\top\\
a_k^{0\to 2n-1}&\boxed{0}
\end{array}
\right|=-\phi_{k+j}^\top,\quad 0\leq k\leq 2n-1,
\end{align}
and thus
\begin{align*}
\sum_{k=1}^{2n}\left|\begin{array}{cc}
A&e_k\\
a_i^{0\to 2n-1}&\boxed{0}
\end{array}
\right|\cdot\left|\begin{array}{cc}
A&\Phi_j^\top\\
a_k^{0\to 2n-1}&\boxed{0}
\end{array}
\right|+\left|\begin{array}{cc}
A&\Phi_{j+1}^\top\\
a_i^{0\to 2n-1}&\boxed{0}
\end{array}
\right|=\Pf(\bullet,\boxed{i,c_{2n-1}})\Pf(\bullet,\boxed{2n,d_j}).
\end{align*}
\end{proof}

\subsection{Derivatives of Grammian type quasi-Pfaffians}
In this part, we focus on derivative formulae for Grammian-type Pfaffians, which are defined as quasi-Pfaffians with elements satisfying
\begin{align}\label{time}
\begin{aligned}
&\frac{d}{dt}\Pf(i,j)=\phi_{i+1}^\top\phi_j-\phi_i^\top\phi_{j+1}=\Pf(i,d_1)\Pf(j,d_0)^\top-\Pf(i,d_0)\Pf(j,d_1)^\top,\\
& \frac{d}{dt}\Pf(i,d_j)=\phi_{i+j+1}^\top=\Pf(i,d_{j+1}).
\end{aligned}
\end{align}
In this case, we can state the following theorem.
\begin{theorem}\label{thm3.4}
For arbitrary $i,j\in\mathbb{N}$ and  $\bullet=\{0,\cdots,2n-1\}$, one has
\begin{align}\label{gram1}
\frac{d}{dt}\Pf(\bullet,\boxed{i,j})=\Pf(\bullet,\boxed{i,d_1})\cdot\Pf(\bullet,\boxed{j,d_0})^\top-\Pf(\bullet,\boxed{i,d_0})\cdot\Pf(\bullet,\boxed{j,d_1})^\top.
\end{align}
\end{theorem}

\begin{proof}
Following the notation above, we have
\begin{align*}
\frac{d}{dt}\Pf(\bullet ,\boxed{i,j})&=\phi_{i+1}^\top\phi_j-\phi_i^\top\phi_{j+1}-\phi_{i+1}^\top \left|\begin{array}{cc}
A&a^j_{0\to 2n-1}\\
-\Phi_0&\boxed{0}
\end{array}
\right|+\phi_i^\top \left|\begin{array}{cc}
A&a^j_{0\to 2n-1}\\
-\Phi_1&\boxed{0}
\end{array}\right|\\
&- \left|\begin{array}{cc}
A&\Phi_1^\top\\
a_i^{0\to 2n-1}&\boxed{0}
\end{array}\right|\cdot \left|\begin{array}{cc}
A&a^j_{0\to 2n-1}\\
-\Phi_0&\boxed{0}
\end{array}\right|+\left|\begin{array}{cc}
A&\Phi_0^\top\\
a_i^{0\to 2n-1}&\boxed{0}
\end{array}\right|\cdot 
\left|\begin{array}{cc}
A&a^j_{0\to 2n-1}\\
-\Phi_1&\boxed{0}
\end{array}\right|\\
&+\left|\begin{array}{cc}
A&\Phi_1^\top\\
a_i^{0\to 2n-1}&\boxed{0}
\end{array}\right|\phi_j-\left|\begin{array}{cc}
A&\Phi_0^\top\\
a_i^{0\to 2n-1}&\boxed{0}
\end{array}\right|\phi_{j+1}.
\end{align*}
Applying \eqref{pfd0}, we simplify the above equation to
\begin{align*}
&-\left|\begin{array}{cc}
A&\Phi_1^\top\\
a_i^{0\to 2n-1}&\boxed{\phi_{i+1}^\top}
\end{array}\right|\cdot\left|\begin{array}{cc}
A&a^j_{0\to 2n-1}\\
-\Phi_0&\boxed{-\phi_j}
\end{array}\right|+\left|\begin{array}{cc}
A&\Phi_0^\top\\
a_i^{0\to 2n-1}&\boxed{\phi_i^\top}
\end{array}\right|\cdot\left|\begin{array}{cc}
A&a^j_{0\to 2n-1}\\
-\Phi_1&\boxed{-\phi_{j+1}}
\end{array}\right|\\
&=-\Pf(\bullet,\boxed{i,d_1})\cdot\Pf(\bullet,\boxed{d_0,j})+\Pf(\bullet,\boxed{i,d_0})\cdot\Pf(\bullet,\boxed{d_1,j}).
\end{align*}
We complete our proof by recognizing the relation \eqref{rel1}.
\end{proof}

\begin{remark}
This formula also has a commutative counterpart. In \cite[Section 2.11]{hirota04}, it has been shown that if Pfaffian elements satisfy a Grammian type derivative formula
\begin{align*}
\frac{d}{dt}\Pf(i,j)=\Pf(d_0,d_1,i,j), \quad \Pf(d_0,d_1)=0,
\end{align*}
then there holds
\begin{align*}
\frac{d}{dt}\Pf(0,\cdots,2n-1)=\Pf(d_0,d_1,0,\cdots,2n-1).
\end{align*}
Theorem \ref{thm3.4} is a generalization of this formula. By taking $\bullet=\{0,\cdots,2n-1\}$ and $i=2n$, $j=2n+1$, then the left hand side in \eqref{gram1} is equivalent to \begin{align*}\frac{d}{dt}\left(\frac{\Pf(\bullet,2n,2n+1)}{\Pf(\bullet)}\right)
=\frac{\Pf(\bullet)\frac{d}{dt}\Pf(\bullet,2n,2n+1)-\Pf(\bullet,2n,2n+1)\frac{d}{dt}\Pf(\bullet)}{\Pf(\bullet)^2},
\end{align*}
while the right hand side is equal to 
\begin{align*}
&\frac{1}{\Pf(\bullet)^2}\left(
\Pf(d_0,\bullet,2n)\Pf(d_1,\bullet,2n+1)-\Pf(d_0,\bullet,2n+1)\Pf(d_1,\bullet,2n)
\right)\\
&\quad=\frac{1}{\Pf(\bullet)^2}\left(\Pf(d_0,d_1,\bullet,2n-2,2n-1)\Pf(\bullet)-\Pf(d_0,d_1,\bullet)\Pf(\bullet,2n-2,2n-1)\right),
\end{align*}
where the last equality is based on the Pfaffian identity. A comparison of each term leads to the derivative formula in commutative case.
\end{remark}

\begin{theorem}\label{thm3.6}
For $\bullet=\{0,\cdots,2n-1\}$ and arbitrary $i\in\mathbb{N}$, one has
\begin{align}\label{gram2}
\frac{d}{dt}\Pf(\bullet,\boxed{i,d_0})=\Pf(\bullet,\boxed{i,d_1})\left(
1-\Pf(\bullet,\boxed{d_0,d_0})
\right)+\Pf(\bullet,\boxed{i,d_0})\cdot\Pf(\bullet,\boxed{d_1,d_0}),
\end{align}
where $\Pf(d_i,d_j)=0$ for any $i,j\in\mathbb{N}$.
\end{theorem}
\begin{proof}
It is noted that 
\begin{align*}
\frac{d}{dt}\Pf(\bullet,\boxed{i,d_0})=\phi_{i+1}^\top-\frac{d}{dt} (a_i^{0\to 2n-1}A^{-1}\Phi_0^\top),
\end{align*}
and the derivative formulas for $a_i^{0\to 2n-1}$, $A$ and $\Phi_i$ satisfy
\begin{align*}
\frac{d}{dt} a_i^{0\to 2n-1}=\phi_{i+1}^\top\Phi_0-\phi_i^\top\Phi_1,\quad \frac{d}{dt} A=\Phi_1^\top\Phi_0-\Phi_0^\top\Phi_1,\quad \frac{d}{dt} \Phi_i=\Phi_{i+1}.
\end{align*}
It leads to the result
\begin{align*}
\frac{d}{dt}\Pf(\bullet,\boxed{i,d_0})&=\phi_{i+1}^\top-\phi_{i+1}^\top\left|
\begin{array}{cc}
A&\Phi_0^\top\\
-\Phi_0&\boxed{0}
\end{array}
\right|+\phi_i^\top\left|
\begin{array}{cc}
A&\Phi_0^\top\\
-\Phi_1&\boxed{0}
\end{array}
\right|
+\left|
\begin{array}{cc}
A&\Phi_1^\top\\
a_i^{0\to 2n-1}&\boxed{0}
\end{array}
\right|
\\
&-\left|
\begin{array}{cc}
A&\Phi_1^\top\\
a_i^{0\to 2n-1}&\boxed{0}
\end{array}
\right|\left|
\begin{array}{cc}
A&\Phi_0^\top\\
-\Phi_0&\boxed{0}
\end{array}
\right|+\left|
\begin{array}{cc}
A&\Phi_0^\top\\
a_i^{0\to 2n-1}&\boxed{0}
\end{array}
\right|\left|
\begin{array}{cc}
A&\Phi_0^\top\\
-\Phi_1&\boxed{0}
\end{array}
\right|\\
&=\left|
\begin{array}{cc}
A&\Phi_1^\top\\
a_i^{0\to 2n-1}&\boxed{\phi_{i+1}^\top}
\end{array}
\right|\left(-\left|
\begin{array}{cc}
A&\Phi_0^\top\\
-\Phi_0&\boxed{0}
\end{array}
\right|+1\right)+\left|
\begin{array}{cc}
A&\Phi_0^\top\\
a_i^{0\to 2n-1}&\boxed{\phi_i^\top}
\end{array}
\right|\left|
\begin{array}{cc}
A&\Phi_0^\top\\
-\Phi_1&\boxed{0}
\end{array}
\right|.
\end{align*}
\end{proof}
\begin{remark}
In the commutative case, this formula is equivalent to the derivative formulae \cite[Eqs. (4.98b) \& (4.98d)]{hirota04}.
The left hand side in the above formula is equal to
\begin{align*}
\frac{d}{dt}
\left(
\frac{\Pf(\bullet,2n,d_0)}{\Pf(\bullet)}
\right)
=
\frac{1}{\Pf(\bullet)^2}     
 \left(
   \Pf(\bullet) \frac{d}{dt}\Pf(\bullet,2n,d_0)-\Pf(\bullet,2n,d_0)\frac{d}{dt}\Pf(\bullet)
\right),
\end{align*}
while the right hand side is equal to 
\begin{align*}
\frac{\Pf(\bullet,2n,d_1)}{\Pf(\bullet)}+\frac{\Pf(\bullet,2n,d_0)\Pf(\bullet,d_1,d_0)}{\Pf(\bullet)^2},
\end{align*}
from which we can deduce
\begin{align*}
\frac{d}{dt}\Pf(\bullet,2n,d_0)=\Pf(\bullet,2n,d_1).
\end{align*}
This formula has been widely applied to the construction of B\"acklund transformation and Toda-type equations for BKP hierarchy. 
\end{remark}
An analogous argument to the proof of theorem \ref{thm3.6} leads to the following proposition.
\begin{proposition}
It holds that
\begin{align}\label{gram3}
\begin{aligned}
\frac{d}{dt}\Pf(\bullet,\boxed{d_0,d_0})&=\Pf(\bullet,\boxed{d_1,d_0})\cdot\Pf(\bullet,\boxed{d_0,d_1})-\Pf(\bullet,\boxed{d_0,d_1})\cdot\Pf(\bullet,\boxed{d_1,d_0})\\
&=\left[\Pf(\bullet,\boxed{d_0,d_1}),\Pf(\bullet,\boxed{d_0,d_1})\right],
\end{aligned}
\end{align}
where the bracket is defined by $[a,b]=ab^\top-ba^\top$.
\end{proposition}
In fact, the above equation is a trivial equation in the commutative case. However, in the non-commutative circumstance, the anti-commutator may not vanish. 
We conclude by addressing the derivative formulae for the following Grammian-type quasi-Pfaffians, which are used later.
\begin{coro}
For any $i\in\mathbb{N}$, there holds that
\begin{subequations}
\begin{align}
&\frac{d}{dt}\Pf(\bullet,\boxed{i,c_{2n-1}})=-\Pf(\bullet,\boxed{i,d_1})\Pf(\bullet,\boxed{c_{2n-1},d_0})^\top+\Pf(\bullet,\boxed{i,d_0})\Pf(\bullet,\boxed{c_{2n-1},d_1})^\top,\label{gram4}\\
&\frac{d}{dt}\Pf(\bullet,\boxed{c_{2n-1},d_0})=\Pf(\bullet,\boxed{c_{2n-1},d_1})\left(1-\Pf(\bullet,\boxed{d_0,d_0})\right)+\Pf(\bullet,\boxed{c_{2n-1},d_0})\Pf(\bullet,\boxed{d_1,d_0}).\label{gram5}
\end{align}
\end{subequations}
\end{coro}

\section{On a non-commutative B-Toda lattice equation}\label{sec4}
This section is to demonstrate that the solution of some non-commutative integrable systems could be expressed by quasi-Pfaffians. 
\subsection{B-Toda lattice and its Pfaffian solution}
While the survey \cite{hirota00} presented various commutative coupled and discrete equations solvable by Pfaffians, we narrow our focus to the so-called Toda equation of BKP type (B-Toda equation for short) with the form
\begin{align}\label{b-toda}
\mathcal{D}_t^2\tau_n\cdot\tau_n=2\mathcal{D}_t\tau_{n-1}\cdot\tau_{n+1}.
\end{align}
This commutative equation has been extensively studied. Its bilinear B\"acklund transformation was established in \cite{gilson03}, while its molecule solutions and wave functions were constructed in \cite{chang18}. Furthermore, a reduction from the 2d-Toda hierarchy was presented in \cite{li22}, leading to the derivation of a bilinear hierarchy. It is also noteworthy that the equation is connected to a condensation algorithm for Pfaffians, which is the simplest equation in this hierarchy \cite{li20}.  We first state the following proposition regarding  the Pfaffian solution for B-Toda equation. 
\begin{proposition}[\cite{chang18}]
The solution of B-Toda equation \eqref{b-toda} can be expressed by
\begin{align*}
\tau_{2n}=\Pf(0,1,\cdots,2n-1),\quad \tau_{2n+1}=\Pf(d_0,0,1,\cdots,2n),
\end{align*}
with Pfaffian entries satisfying 
\begin{align*}
&\frac{d}{dt}\Pf(d_0,i)=\Pf(d_0,i+1)=\Pf(d_1,i),\\
& \frac{d}{dt} \Pf(i,j)=\Pf(i+1,j)+\Pf(i,j+1)=\Pf(d_0,d_1,i,j).
\end{align*}
\end{proposition}
This proposition reveals that the B-Toda equation is intrinsically a coupled system
\begin{subequations}
\begin{align}
&\mathcal{D}_t^2\tau_{2n}\cdot\tau_{2n}=2\mathcal{D}_t\tau_{2n-1}\cdot\tau_{2n+1},\label{b-toda1}\\
&\mathcal{D}_t^2\tau_{2n+1}\cdot\tau_{2n+1}=2\mathcal{D}_t\tau_{2n}\cdot\tau_{2n+2}.\label{b-toda2}
\end{align}
\end{subequations}
In contrast to the Pfaff lattice hierarchy studied by Adler et al. \cite{adler02}, the B-Toda equation involves two distinct families of Pfaffian tau-functions.
The lattice equation \eqref{b-toda} is the simplest integrable system with Pfaffian tau-functions, serving the same role as the Toda equation in the Toda/KP hierarchy.

\subsection{Non-commutative B-Toda equation and its quasi-Pfaffian solution}
We now introduce a non-commutative analog of the B-Toda lattice \eqref{b-toda}, anticipating solutions in terms of quasi-Pfaffians. 
Since the conventional bilinear formalism breaks down in the non-commutative setting, our strategy is to start from quasi-Pfaffians and construct the associated lattice equations, thereby deriving the non-commutative integrable system.
First of all, we need to introduce some variables 
\begin{align*}
&u_n=\Pf(\bullet,\boxed{d_1,d_0}), \quad\quad\,\,\,\, \sigma_n=\Pf(\bullet,\boxed{2n,d_0}),\quad \quad\,\,\,\,\tilde{\sigma}_n=\Pf(\bullet,\boxed{2n,d_1}),\\
&\hat{\sigma}_n=\Pf(\bullet,\boxed{2n+1,d_0}),\quad \tilde{\hat{\sigma}}_n=\Pf(\bullet,\boxed{2n+1,d_1}),
\end{align*}
in which $\bullet=\{0,\cdots,2n-1\}$.
The following proposition illustrates the relationship between the variables.
\begin{proposition}
The following equations hold
\begin{subequations}
\begin{align}
u_{n+1}&=u_n+(\tilde{\sigma}_n^\top A_n+\tilde{\hat{\sigma}}_n^\top B_n){\sigma}_n+(\tilde{\sigma}_n^\top C_n+\tilde{\hat{\sigma}}_n^\top D_n)\tilde{\hat{\sigma}}_n,\label{nceq1}\\
\tilde{\sigma}_n&=\left(\frac{d}{dt} \sigma_n-\sigma_nu_n^\top\right)\left(
1-\int[u_n,u_n]dt
\right)^{-1},\label{nceq2}\\
 \tilde{\hat{\sigma}}_n&=\left(\frac{d}{dt}\hat{\sigma}_n-\hat{\sigma}_nu_n^\top\right)\left(
1-\int[u_n,u_n]dt
\right)^{-1},\label{nceq3}
\end{align}
\end{subequations}
where $A_n,\,B_n,\,C_n,\,D_n$ are given by
\begin{align}\label{abcd}
\left(\begin{array}{cc}
A_n&C_n\\
B_n&D_n
\end{array}
\right)=\left(\begin{array}{cc}
\int[\tilde{\sigma}_n,\sigma_n]dt&\int (\tilde{\sigma}_n\hat{\sigma}_n^\top-\sigma_n\tilde{\hat{\sigma}}_n^\top) dt\\
-\int(\hat{\sigma}_n\tilde{\sigma}_n^\top-\tilde{\hat{\sigma}}_n\sigma_n^\top)dt&\int [\tilde{\hat{\sigma}}_n,\hat{\sigma}_n] dt
\end{array}
\right)^{-1}.
\end{align}
\end{proposition}

\begin{proof}
From the Grammian type derivative formula \eqref{gram1}, we obtain
\begin{align*}
&\Pf(\bullet,\boxed{2n,2n})=\int [\tilde{\sigma}_n,\sigma_n]dt,\\
& \Pf(\bullet,\boxed{2n+1,2n+1})=\int [\tilde{\hat{\sigma}}_n,\hat{\sigma}_n]dt, \\
&\Pf(\bullet,\boxed{2n,2n+1})=\int \tilde{\sigma}_n\hat{\sigma}_n^\top-\sigma_n\tilde{\hat{\sigma}}_n^\top dt.
\end{align*}
Applying the quasi-Pfaffian identity \eqref{identity} yields the equation \eqref{nceq1}.
Furthermore, the Grammian-type derivative formula \eqref{gram3} reveals that the variables $\tilde{\sigma}_n$ and $\tilde{\hat{\sigma}}_n$ can be expressed in terms of $u_n$, $\sigma_n$ and $\hat{\sigma}_n$.  By recognizing that
\begin{align}\label{d0d0}
\Pf(\bullet,\boxed{d_0,d_0})=\int[u_n,u_n]dt,
\end{align}
we could obtain \eqref{nceq2} and \eqref{nceq3} directly from \eqref{gram2}.
\end{proof}
A substitution of equations \eqref{nceq2} and \eqref{nceq3} into \eqref{nceq1} demonstrates that $\tilde{\sigma}_n$ and $\tilde{\hat{\sigma}}_n$ are auxiliary variables which can be eliminated. One expects to express $\hat{\sigma}_n$ in terms of $u_n$ and $\sigma_n$. 
The following proposition is a direct consequence of the Wronskian-type derivative formula \eqref{wronskian2} with variables
\begin{align*}
v_n=\Pf(\bullet,\boxed{c_{2n-1},d_0}),\quad r_n=\Pf(\bullet,\boxed{2n,c_{2n-1}}).
\end{align*}
\begin{proposition}
We have the following relations between variables
\begin{align}\label{hatsigma}
\hat{\sigma}_n=\frac{d}{dt}\sigma_n-r_n\sigma_n+\int [\sigma_n,\tilde{\sigma}_n]dt \cdot v_n,
\end{align}
where $r_n$ is expressed in terms of $u_n,\,\sigma_n$ and $v_n$ by
\begin{align}\label{nceq7}
\begin{aligned}
\frac{d}{dt} r_n&=-\tilde{\sigma}_nv_n^\top+\sigma_n\left(1-\int[u_n,u_n]dt\right)^{-1}\left(\frac{d}{dt} v_n^\top-u^\top_nv_n^\top\right)\\
&=\left(\sigma_nu_n^\top-\frac{d}{dt}\sigma_n\right)\left(
1-\int[u_n,u_n]dt
\right)^{-1}v_n^\top+\sigma_n\left(
1-\int[u_n,u_n]dt
\right)^{-1}\left(\frac{d}{dt} v_n^\top-u_n^\top v_n^\top\right),
\end{aligned}
\end{align}
while $v_n$ satisfies a recurrence relation
\begin{align}\label{nceq6}
v_{n+1}=(D_nr_n-B_n)\sigma_n-D_n\frac{d}{dt} \sigma_n+\int [\sigma_n,\tilde{\sigma}_n]dt\cdot v_n,
\end{align}
where $B_n$ and $D_n$ is given by \eqref{abcd}.
\end{proposition}
\begin{proof}
First, it is obvious that equation \eqref{hatsigma} is directly obtained by \eqref{wronskian2}. Then following \eqref{gram5}, we get 
\begin{align*}
\Pf(\bullet,\boxed{c_{2n-1},d_1})=\left(\frac{d}{dt} v_n-v_nu_n\right)\left(1-\int [u_n,u_n]dt\right)^{-1},
\end{align*}
from which we can express $r_n$ by using \eqref{gram4}. Moreover, by applying the quasi-Pfaffian identity to $\Pf(\bullet,2n,2n+1,\boxed{c_{2n+1},d_0})$, we get 
\begin{align*}
v_{n+1}=-(0,1)\left(\begin{array}{cc}
A_n&C_n\\
B_n&D_n
\end{array}
\right)\left(
\sigma_n \atop \hat{\sigma}_n
\right)=-B_n\sigma_n-D_n\hat{\sigma}_n.
\end{align*}
Substitution of $\hat{\sigma}_n$ into the above equation results in the desired formula.
\end{proof}
Similar to the equation \eqref{hatsigma}, 
we have an alternative formula for $\tilde{\hat{\sigma}}_n$, namely, 
\begin{align}\label{nceq4}
\tilde{\hat{\sigma}}_n=\frac{d}{dt} \tilde{\sigma}_n-r_n\tilde{\sigma}_n+\int[\sigma_n,\tilde{\sigma}_n]dt\cdot v_n.
\end{align}
Therefore, from the compatibility condition of equations \eqref{nceq3} and \eqref{nceq4}, we have the following equation.
\begin{proposition}
The compatibility condition yields the equation
\begin{align}\label{nceq5}
\begin{aligned}
&\left(
-\frac{d}{dt} r_n\sigma_n+\frac{d}{dt}\left(
\int[\sigma_n,\tilde{\sigma}_n]dt\cdot v_n
\right)+\sigma_n\frac{d}{dt} u_n^\top
\right)\left(
1-\int[u_n,u_n]dt
\right)^{-1}\\
&\quad=\left(\frac{d}{dt} \sigma_n-\sigma_nu_n^\top\right)\frac{d}{dt} \left(1-\int [u_n,u_n]dt\right)^{-1}+\int[\sigma_n,\tilde{\sigma}_n]dt\left(1+u_n^\top (1-\int[u_n,u_n]dt)^{-1}\right).
\end{aligned}
\end{align}
 \end{proposition}
Therefore, the non-commutative B-Toda lattice is described by the closed system defined by  $u_n,\,\sigma_n,\,v_n$ and $r_n$, and equations \eqref{nceq1}, \eqref{nceq7}, \eqref{nceq6} and \eqref{nceq5}. In the commutative case, these variables can be expressed in terms of Pfaffian $\tau$-functions via relations
\begin{align*}
&\sigma_n=\frac{\tau_{2n+1}}{\tau_{2n}},\quad v_n=\frac{\tau_{2n-1}}{\tau_{2n}}, \quad u_n=r_n=-\frac{\frac{d}{dt} \tau_{2n+1}}{\tau_{2n}},\\
&\tilde{\sigma}_n=\hat{\sigma}_n=\frac{\frac{d}{dt} \tau_{2n+1}}{\tau_{2n}},\quad \tilde{\hat{\sigma}}_n=\frac{\frac{d^2}{dt^2}\tau_{2n+1}}{\tau_{2n}},
\end{align*}
and
\begin{align*}
\left(\begin{array}{cc}
A_n&C_n\\
B_n&D_n
\end{array}
\right)=\left(
\begin{array}{cc}
0&{\tau_{2n+2}}{\tau_{2n}^{-1}}\\
-{\tau_{2n+2}}{\tau_{2n}^{-1}}&0
\end{array}
\right)^{-1}=\left(
\begin{array}{cc}
0&-{\tau_{2n}}{\tau_{2n+2}^{-1}}\\
{\tau_{2n}}{\tau_{2n+2}^{-1}}&0
\end{array}
\right).
\end{align*}
Substituting these commutative expressions into the system, one finds that equations \eqref{nceq1} and \eqref{nceq7} reduce to equations \eqref{b-toda1} and \eqref{b-toda2} respectively, while equations \eqref{nceq6} and \eqref{nceq5} become trivial. The non-commutative equation is more complicated 
than the commutative one, as there are trivial equations needed to get a closed system.

\subsection{Lax pair of the system}
Orthogonal polynomial theory serves as a powerful tool for constructing Lax pairs for these systems. 
Notably, the Pfaff lattice hierarchy and the B-Toda lattice hierarchy have been derived from the theories of skew-orthogonal polynomials \cite{adler99} and partial-skew-orthogonal polynomials \cite{chang18}, respectively. Non-commutative integrable systems and matrix-valued orthogonal polynomials has also been extensively studied \cite{gilson25,li242,li252,doliwa25}. 
In this paper, we construct the Lax pair of non-commutative B-Toda lattice by introducing $\mathcal{R}$-valued partial-skew-orthogonal polynomials. 

Let us define an $\mathcal{R}$-valued skew-symmetric inner product
\begin{align}\label{sip}
\langle x^i,y^j\rangle=\int_{\mathbb{R}_+^2}\frac{x-y}{x+y}x^iy^jW^\top(x)W(y)dxdy
\end{align}
with an $\mathcal{R}$-valued weight function $W:\mathbb{R}_+\to\mathcal{R}$. Clearly, defining $a_{i,j}=\langle x^i,y^j\rangle$ implies that $a_{i,j}^\top=-a_{j,i}$. Furthermore, introducing the time-deformed weight function $W(x;t)=\exp(xt)W(x)$ leads to the time evolutions \eqref{time} for $a_{i,j}(t)$, where $\phi_i(t)=\int_{\mathbb{R}_+}x^iW(x)\exp(xt)dx$.

\begin{definition}
With the skew symmetric inner product \eqref{sip}, we can define a family of monic polynomials with even order $\{P_{2n}(x;t)\}_{n\in\mathbb{N}}$ by skew-orthogonality
\begin{align}\label{sop1}
\langle P_{2n}(x;t),x^i\rangle=0,\quad i=0,\cdots,2n-1.
\end{align}
\end{definition}
This family of polynomials is well defined. Assuming 
\begin{align*}
P_{2n}(x;t)=x^{2n}+\xi_{2n,2n-1}x^{2n-1}+\cdots+\xi_{2n,0},\quad \xi_{2n,i}\in\mathcal{R},
\end{align*}
the skew-orthogonality condition \eqref{sop1} yields the linear system
\begin{align*}
a_{2n,i}+\xi_{2n,2n-1}a_{2n-1,i}+\cdots+\xi_{2n,0}a_{0,i}=0,\quad i=0,1,\cdots,2n-1.
\end{align*}
Then, solving the linear system with non-commutative coefficients leads to a quasi-Pfaffian expression for the coefficients, resulting in (see appendix)
\begin{align*}
\xi_{2n,i}=\Pf(0,\cdots,2n-1,\boxed{2n,c_i}).
\end{align*} 
To conclude, we have the following proposition for these even order polynomials.
\begin{proposition}
$\{P_{2n}(x;t)\}_{n\in\mathbb{N}}$ satisfying \eqref{sop1} have the following quasi-Pfaffian expressions
\begin{align*}
P_{2n}(x;t)=\Pf(0,\cdots,2n-1,\boxed{2n,x})=\left|
\begin{array}{cccc}
a_{0,0}&\cdots&a_{0,2n-1}&1\\
\vdots&&\vdots&\vdots\\
a_{2n-1,0}&\cdots&a_{2n-1,2n-1}&x^{2n-1}\\
a_{2n,0}&\cdots&a_{2n,2n-1}&\boxed{x^{2n}}
\end{array}
\right|.
\end{align*}
\end{proposition}
In the commutative case, odd-order polynomials associated with a skew-symmetric inner product are not well-defined due to the singularity of odd-sized skew-symmetric matrices.
However, this singularity can be circumvented by modifying the skew-orthogonality. To construct the Lax pair for the non-commutative B-Toda lattice, we consider a family of odd-order polynomials $\{P_{2n+1}(x;t)\}_{n\in\mathbb{N}}$ satisfying
\begin{align}\label{sop2}
\langle P_{2n+1}(x;t),x^i\rangle=-\phi_i,\quad i=0,1,\cdots,2n+1.
\end{align}
By assuming that
\begin{align*}
P_{2n+1}(x;t)=\xi_{2n+1,2n+1}x^{2n+1}+\xi_{2n+1,2n}x^{2n}+\cdots+\xi_{2n+1,0},
\end{align*}
these coefficients could be solved as 
\begin{align*}
\xi_{2n+1,i}=\Pf(0,\cdots,2n+1,\boxed{d_0,c_i}).
\end{align*}
Moreover, we have the following proposition.
\begin{proposition}
$\{P_{2n+1}(x;t)\}_{n\in\mathbb{N}}$ satisfying \eqref{sop2} have the following quasi-Pfaffian expressions
\begin{align*}
P_{2n+1}(x;t)=\Pf(0,\cdots,2n+1,\boxed{d_0,x})=\left|\begin{array}{cccc}
a_{0,0}&\cdots&a_{0,2n+1}&1\\
\vdots&&\vdots&\vdots\\
a_{2n+1,0}&\cdots&a_{2n+1,2n+2}&x^{2n+1}\\
-\phi_0&\cdots&-\phi_{2n+1}&\boxed{0}
\end{array}
\right|.
\end{align*}
where $\Pf(d_0,x)=0$.
\end{proposition}

The following theorem is a derivative formula for these polynomials.
\begin{theorem}\label{sopthm1}
There exists the following derivative formula
\begin{align}\label{t1}
\partial_t P_{2n}(x;t)+\alpha_n\partial_tP_{2n-1}(x;t)=\beta_nP_{2n-1}(x;t),
\end{align}
where 
\begin{align*}
\alpha_n=-\sigma_n\left(1-\int[u_n,u_n]dt\right)^{-1},\quad \beta_n=\tilde{\sigma}_n+\sigma_n\left(1-\int[u_n,u_n]dt\right)^{-1}u_n^\top.
\end{align*}
\end{theorem}
The proof of this theorem employs the following lemma and a family of auxiliary polynomials.
\begin{lemma}
For polynomials $\{P_{2n}(x;t)\}_{n\in\mathbb{N}}$, we have the following derivative formula \begin{align}
\partial_t P_{2n}(x;t)=\tilde{\sigma}_{n}P_{2n-1}(x;t)-\sigma_{n}\tilde{P}_{2n-1}(x;t),\label{sopd1}
\end{align}
where $\tilde{P}_{2n-1}(x;t)=\Pf(0,\cdots,2n-1,\boxed{d_1,x})$. Moreover, we have
\begin{align}\label{sopd2}
\partial_t P_{2n-1}(x;t)=-u_n^\top P_{2n-1}(x;t)-\left(1-\int[u_n,u_n]dt\right)\tilde{P}_{2n-1}(x;t).
\end{align}
\end{lemma}
\begin{proof}
This proof is based on the derivative formula for Grammian type elements. 
By introducing the vector $\chi_{2n-1}(x)=(1,x,\cdots,x^{2n-1})^\top$, we have the derivative formula
\begin{align*}
\pt P_{2n}(x;t)&=\pt\left(
x^{2n}-a_{2n}^{0\to 2n-1}A^{-1}\chi_{2n-1}(x)
\right)\\
&=-(\phi_{2n+1}^\top\Phi_0-\phi_{2n}^\top\Phi_1)A^{-1}\chi_{2n-1}(x)+a_{2n}^{0\to 2n-1}A^{-1}(\Phi_1^\top\Phi_0-\Phi_0^\top\Phi_1)A^{-1}\chi_{2n-1}(x)\\
&=\Pf(0,\cdots,2n-1,\boxed{2n,d_1})P_{2n-1}(x;t)-\Pf(0,\cdots,2n-1,\boxed{2n,d_0})\tilde{P}_{2n-1}(x;t),
\end{align*}
which is exactly the equation \eqref{sopd1}. Moreover, if we consider the derivative of $P_{2n-1}(x;t)$, we have
\begin{align*}
\pt P_{2n-1}(x;t)&=\pt\left(
\Phi_0A^{-1}\chi_{2n-1}(x)
\right)\\
&=\Phi_1A^{-1}\chi_{2n-1}(x)-\Phi_0A^{-1}(\Phi_1^\top\Phi_0-\Phi_0^\top\Phi_1)A^{-1}\chi_{2n-1}(x)
\\
&=-\Pf(0,\cdots,2n-1,\boxed{d_0,d_1})P_{2n-1}(x;t)-(1-\Pf(0,\cdots,2n-1,\boxed{d_0,d_0}))\tilde{P}_{2n-1}(x;t).
\end{align*}
An application of equation \eqref{d0d0} results in the formula.
\end{proof}

On the other hand, we have the following spectral problems, as generalizations of the equations \cite[Equation 3.7]{adler99} and \cite[Equation 3.15]{li24}.
\begin{theorem}\label{sopthm2}
We have the spectral problems
\begin{subequations}
\begin{align}
(\pt+x)P_{2n}(x;t)&=a_n^{-1}P_{2n+1}(x;t)+(r_n-a_n^{-1}b_n)P_{2n}(x;t)\label{spec1}\\
&-(a_n^{-1}+\int[\sigma_n,\tilde{\sigma}_n]dt D_{n-1}a_{n-1}^{-1})P_{2n-1}(x;t)\nonumber\\
&+\int [\sigma_n,\tilde{\sigma}_n]dt (B_{n-1}-D_{n-1}a_{n-1}^{-1}b_{n-1})P_{2n-2}(x;t)\nonumber\\
&-\int[\sigma_n,\tilde{\sigma}_n]dt D_{n-1}a_{n-1}^{-1} P_{2n-3}(x;t),\nonumber\\
(\pt+x)P_{2n-1}(x;t)&=v_n P_{2n}(x;t)+\sigma_nD_{n-1}a_{n-1}^{-1}P_{2n-1}(x;t)\label{spec2}\\
&+\sigma_n(B_{n-1}-D_{n-1}a_{n-1}^{-1}b_{n-1})P_{2n-2}(x;t)-\sigma_n D_{n-1}a_{n-1}^{-1}P_{2n-3}(x;t),\nonumber
\end{align}
\end{subequations}
where
\begin{align*}
a_n=\sigma_n^\top C_n+\hat{\sigma}_n^\top D_n,\quad b_n=\sigma_n^\top A_n+\hat{\sigma}_n^\top B_n.
\end{align*}
\end{theorem}

We prove this theorem by making use of the following facts. 
\begin{lemma}\label{soplem1}
We claimed that 
\begin{align*}
(\pt+x)P_{2n}(x;t)=Q_{2n+1}(x;t)+r_nP_{2n}(x;t)-\int[\sigma_n,\tilde{\sigma}_n]dt\cdot \tilde{Q}_{2n-1}(x;t),
\end{align*}
where 
\begin{align*}
Q_{2n+1}(x;t)=\Pf(0,\cdots,2n-1,\boxed{2n+1,x}),\quad \tilde{Q}_{2n-1}(x;t)=\Pf(0,\cdots,2n-1,\boxed{c_{2n-1},x}).
\end{align*}
In the above notations, we require that $\Pf(c_{2n-1},x)=0$. 
\end{lemma}
\begin{proof}
Noting that
\begin{align*}
\pt P_{2n}(x;t)&=\pt \left(x^{2n}-
a_{2n}^{0\to 2n-1}A^{-1}\chi_{2n-1}(x)
\right)\\
&=-a_{2n+1}^{0\to 2n-1}A^{-1}\chi_{2n-1}(x)-a_{2n}^{1\to 2n}A^{-1}\chi_{2n-1}(x)+a_{2n}^{0\to 2n-1}A^{-1}(\pt A)A^{-1}\chi_{2n-1}(x)\\
&=-x^{2n+1}+\Pf(0,\cdots,2n-1,\boxed{2n+1,x})+\left|\begin{array}{cc}
A&\chi_{2n-1}(x)\\
a_{2n}^{1\to 2n}&\boxed{0}
\end{array}
\right|\\
&+\sum_{k=1}^{2n}\left|\begin{array}{cc}
A&a^{k}_{0\to 2n-1}\\
a_{2n}^{0\to 2n-1}&\boxed{0}
\end{array}
\right|\cdot\left|\begin{array}{cc}
A&\chi_{2n-1}(x)\\
e_k^\top&\boxed{0}
\end{array}
\right|\\
&+\sum_{k=1}^{2n}
\left|\begin{array}{cc}
A&e_k\\
a_{2n}^{0\to 2n-1}&\boxed{0}
\end{array}
\right|\cdot\left|\begin{array}{cc}
A&\chi_{2n-1}(x)\\
a_{k}^{0\to 2n-1}&\boxed{0}
\end{array}
\right|,
\end{align*}
and from the zero condition \eqref{zerocondition}, we have
\begin{align*}
\left|\begin{array}{cc}
A&a^{k}_{0\to 2n-1}\\
a_{2n}^{0\to 2n-1}&\boxed{0}
\end{array}
\right|=-a_{2n,k},\quad
\left|\begin{array}{cc}
A&\chi_{2n-1}(x)\\
a_{k}^{0\to 2n-1}&\boxed{0}
\end{array}
\right|=-x^{k},\quad
 1\leq k\leq 2n-1.
\end{align*}
Therefore,
\begin{align*}
&\sum_{k=1}^{2n}\left|\begin{array}{cc}
A&a^{k}_{0\to 2n-1}\\
a_{2n}^{0\to 2n-1}&\boxed{0}
\end{array}
\right|\cdot\left|\begin{array}{cc}
A&\chi_{2n-1}(x)\\
e_k^\top&\boxed{0}
\end{array}
\right|\\
&\quad=-\left|\begin{array}{cc}
A&\chi_{2n-1}(x)\\
a_{2n}^{1\to 2n}&\boxed{0}
\end{array}
\right|-\Pf(0,\cdots,2n-1,\boxed{2n,2n})\tilde{Q}_{2n-1}(x;t),
\end{align*}
and 
\begin{align*}
&\sum_{k=1}^{2n}
\left|\begin{array}{cc}
A&e_k\\
a_{2n}^{0\to 2n-1}&\boxed{0}
\end{array}
\right|\cdot\left|\begin{array}{cc}
A&\chi_{2n-1}(x)\\
a_{k}^{0\to 2n-1}&\boxed{0}
\end{array}
\right|\\&\quad=-x(-x^{2n}+P_{2n}(x;t))+\Pf(0,\cdots,2n-1,\boxed{2n,c_{2n-1}})P_{2n}(x;t).
\end{align*}
The claim follows from a combination of the preceding formulas.
\end{proof}

In fact, $Q_{2n+1}(x)$ have an intimate relation with the partial-skew-orthogonal polynomials by using the quasi-Pfaffian identity.
\begin{proposition}\label{prop1}
We have the following relationship between $Q_{2n+1}(x;t)$, $\tilde{Q}_{2n-1}(x;t)$ and $P_{n}(x;t)$ via
\begin{align*}
Q_{2n+1}(x;t)=a_n^{-1}\left(P_{2n+1}(x;t)-b_nP_{2n}(x;t)-P_{2n-1}(x;t)\right).
\end{align*}
Besides, we have the equation
\begin{align*}
\tilde{Q}_{2n+1}(x;t)&=D_na_n^{-1}P_{2n+1}(x;t)+(B_n-D_na_n^{-1}b_n)P_{2n}(x;t)-D_na_n^{-1}P_{2n-1}(x;t).
\end{align*}
\end{proposition}
\begin{proof}
Applying the identity \eqref{identity} to the quasi-Pfaffian $\Pf(\bullet,2n,2n+1,\boxed{d_0,x})$, we could obtain
\begin{align*}
P_{2n+1}(x)=P_{2n-1}(x)-\left(
\Pf(\bullet,\boxed{d_0,2n}),\Pf(\bullet,\boxed{d_0,2n+1})
\right)\left(\begin{array}{cc}
A_n&C_n\\
B_n&D_n
\end{array}
\right)\left(\begin{array}{c}
\Pf(\bullet,\boxed{2n,x})\\
\Pf(\bullet,\boxed{2n+1,x})
\end{array}
\right).
\end{align*}
This is the exactly the first equation by using the notation in \eqref{abcd}. Then applying the identity to $\Pf(\bullet,2n,2n+1,\boxed{c_{2n+1},x})$, we could obtain
\begin{align*}
\tilde{Q}_{2n+1}(x;t)&=(0,1)\left(\begin{array}{cc}
A_n&C_n\\
B_n&D_n
\end{array}
\right)\left(\begin{array}{c}
P_{2n}(x;t)\\
Q_{2n+1}(x;t)
\end{array}
\right)=B_nP_{2n}(x;t)+D_nQ_{2n+1}(x;t).
\end{align*}
A substitution of ${Q}_{2n+1}(x;t)$ results in the formula for $\tilde{Q}_{2n+1}(x;t)$.
\end{proof}
Having established equation \eqref{spec1} via Proposition \ref{prop1} and Lemma \ref{soplem1}, we now turn to the spectral problem for odd-order polynomials. By analogy with the proof of Lemma \ref{soplem1}, we obtain the following result.
\begin{lemma}\label{soplem2}
For odd order polynomials $\{P_{2n-1}(x;t)\}_{n\in\mathbb{N}_+}$, we have the formula
\begin{align*}
(\pt+x)P_{2n-1}(x;t)=v_nP_{2n}(x;t)+\sigma_n\tilde{Q}_{2n-1}(x;t).
\end{align*}
\end{lemma}
Therefore, from Proposition \ref{prop1} and Lemma \ref{soplem2}, we could get the equation \eqref{spec2}.
To conclude,  we have the following recurrence relations for polynomials from Theorem \ref{sopthm1} and Theorem \ref{sopthm2}. 
\begin{theorem}
The $\mathcal{R}$-valued polynomials $\{P_n(x;t)\}_{n\in\mathbb{N}}$ satisfy the following recurrence relation
\begin{align}\label{recurrence1}
\begin{aligned}
x(P_{2n}(x;t)&+\alpha_nP_{2n-1}(x;t))=a_n^{-1}P_{2n+1}(x;t)+(r_n-a_n^{-1}b_n+\alpha_nv_n)P_{2n}(x;t)\\
&-(a_n^{-1}+\beta_n+(\int [\sigma_n,\tilde{\sigma}_n]dt-\alpha_n\sigma_n)D_{n-1}a_{n-1}^{-1})P_{2n-1}(x;t)\\
&+(\int [\sigma_n,\tilde{\sigma}_n]dt+\alpha_n\sigma_n)(B_{n-1}-D_{n-1}a_{n-1}^{-1}b_{n-1})P_{2n-2}(x;t)\\
&-(\int [\sigma_n,\tilde{\sigma}_n]dt+\alpha_n\sigma_n)D_{n-1}a_{n-1}^{-1}P_{2n-3}(x;t).
\end{aligned}
\end{align}
\end{theorem}
This recurrence relation constitutes a spectral problem linking $P_{2n}(x;t)$ and $P_{2n-1}(x;t)$. A complementary relation between $P_{2n}(x;t)$ and $P_{2n+1}(x;t)$is therefore required. To derive it, we first establish the following lemma.
\begin{lemma}
There holds
\begin{align*}
\tilde{P}_{2n+1}(x;t)=\tilde{P}_{2n-1}(x;t)+\tilde{b}_n P_{2n}(x;t)+\tilde{a}_n Q_{2n+1}(x;t),
\end{align*}
where $\tilde{a}_n=\tilde{\sigma}_n^\top C_n+\tilde{\hat{\sigma}}_n^\top D_n$ and $\tilde{b}_n=\tilde{\sigma}_n^\top A_n+\tilde{\hat{\sigma}}_n^\top B_n$.
\end{lemma}
Therefore, by recognizing from \eqref{sopd2} that
\begin{align*}
\tilde{P}_{2n+1}(x;t)&=-\left(
1-\int [u_{n+1},u_{n+1}]dt
\right)^{-1}\pt P_{2n+1}(x;t)\\
&\quad-\left(
1-\int [u_{n+1},u_{n+1}]dt
\right)^{-1}u_{n+1}^\top P_{2n+1}(x;t),
\end{align*}
and from \eqref{sopd1} that
\begin{align*}
\tilde{P}_{2n-1}(x;t)=-\sigma_n^{-1}\pt P_{2n}(x;t)+\sigma_n^{-1}\tilde{\sigma}_nP_{2n-1}(x;t),
\end{align*}
we could get 
\begin{align}
\begin{aligned}\label{t2}
\pt P_{2n+1}(x;t)&+\gamma_n \pt P_{2n}(x;t)=-\left(u_{n+1}^\top-
\gamma_n\sigma_n\tilde{a}_na_n^{-1}
\right)P_{2n+1}(x;t)\\&+\gamma_n\sigma_n(\tilde{b}_n-\tilde{a}_na_n^{-1}b_n)P_{2n}(x;t)+\gamma_n(\tilde{\sigma}_n-\sigma_n\tilde{a}_na_n^{-1})P_{2n-1}(x;t),
\end{aligned}
\end{align}
where $\gamma_n=-\left(
1-\int[u_{n+1},u_{n+1}]dt
\right)\sigma_n^{-1}$.
This formula, together with Theorem \ref{sopthm2}, lead to the rest part of the spectral problem.
\begin{theorem}
The $\mathcal{R}$-valued polynomials $\{P_n(x;t)\}_{n\in\mathbb{N}}$ satisfy the following recurrence relation
\begin{align}\label{recurrence2}
\begin{aligned}
x(P_{2n+1}(x;t)&+\gamma_nP_{2n}(x;t))=v_{n+1}P_{2n+2}(x;t)+(u_{n+1}^\top+\kappa_n)P_{2n+1}(x;t)\\
&+\left(
\sigma_{n+1}B_n+\gamma_n(r_n-\sigma_n\tilde{b}_n)-\kappa_nb_n
\right)P_{2n}(x;t)\\
&-\left(
r_n\tilde{\sigma}_n+\int [\sigma_n,\tilde{\sigma}_n]dt D_na_n^{-1}+\kappa_n
\right)P_{2n-1}(x;t)\\
&+\gamma_n\int[\sigma_n,\tilde{\sigma}_n]dt(B_{n-1}-D_{n-1}a_{n-1}^{-1}b_{n-1})P_{2n-2}(x;t)\\
&-\gamma_n\int[\sigma_n,\tilde{\sigma}_n]dtD_{n-1}a_{n-1}^{-1}P_{2n-3}(x;t),
\end{aligned}
\end{align}
where we simply denote $\kappa_n=\sigma_{n+1}D_n+\gamma_n-\gamma_n\sigma_n\tilde{a}_n$.
\end{theorem}
It should be noted that we could write the recurrence relations \eqref{recurrence1} and \eqref{recurrence2} as matrix form $$xL_1P(x;t)=L_2P(x;t),$$ where $P(x)$ is a semi-infinite column vector composed by $P_0(x;t),P_1(x;t),\cdots$. At the same time, the derivative formulas \eqref{t1} and \eqref{t2}
could be simply denoted by $$L_1\pt P(x;t)=L_3P(x;t).$$
The compatibility condition leads to the Lax formalism 
\begin{align*}
\pt(L_1^{-1}L_2)=L_1^{-1}L_3L_1^{-1}L_2-L_1^{-1}L_2L_1^{-1}L_3.
\end{align*}
This is a non-commutative generalization of fractional difference operator considered in \cite{krichever95}.

\section{Concluding remarks}

This paper has introduced a non-commutative generalization of the Pfaffian---the quasi-Pfaffian---defined for matrices over a division ring. Inspired by the theory of quasi-determinants developed by Gelfand and Retakh, this extension enables the analysis of linear systems with non-commutative skew-symmetric coefficients, expressed explicitly in terms of quasi-Pfaffians. We have established fundamental properties of quasi-Pfaffians and demonstrated their utility in the context of non-commutative integrable systems and polynomial theory. 

A natural question arises regarding the non-commutative analogue of the classical identity 
$\det(A)=\Pf(A)^2$. While such a relationship remains an open problem, its resolution would significantly enhance the interplay between quasi-determinants and quasi-Pfaffians. Besides, we mainly consider the matrix elements in quasi-Pfaffian with the anti-involution $\top$, which works similarly as transpose for matrices. However, different from the definition for quantum Pfaffians, 
which is defined on the algebra generated by entries module certain ideal generated by a quadratic relation. The connection established by Krob and Leclerc \cite{krob95} between quasi-determinants and quantum determinants motivates the conjecture that a similarly profound relationship exists between their Pfaffian counterparts.

Future work may also explore deeper connections to combinatorial structures, random processes, and further integrable models in the non-commutative setting. Over the years, Pfaffians with block structures have appeared in diverse studies such as Brownian motions, combinatorics with non-intersecting paths, and so on \cite{li23,li25}. Therefore, it is expected to find more applications of quasi-Pfaffian into different subjects. 

The connection between solutions and polynomials in non-commutative integrable systems, as explored in Section \ref{sec4}, invites further investigation. In the commutative setting, wave functions are known to be constructible from $\tau$-functions \cite{adler02,li22}, suggesting the potential for a matrix decomposition approach to derive dressing operators in the non-commutative case. Previous work on quasi-determinants and matrix-valued orthogonal polynomials has already established relevant decompositions for the non-commutative Toda equation \cite{ariznabarreta14}. Moreover, while a well-defined mapping exists from the commutative Toda lattice to the Pfaff lattice \cite{adler022}, it remains an open and compelling question whether an analogous map connects the non-commutative versions of these hierarchies.

\section*{Acknowledgement}
This work is partially funded by grants (Grant No. NSFC 12371251, 12175155). G.F. Yu was also partially supported by the National Key Research and Development Program of China (Grant No. 2024YFA1014101), Shanghai Frontier Research Institute for Modern Analysis and the Fundamental Research Funds for the Central Universities.

\section*{Conflicts of Interest}
There is no conflict of interest.

\section*{Data Availability Statement}
There is no data needed in this research.

\appendix
\section{Linear system and definition of quasi-Pfaffian}
\subsection{Motivations}\label{sub2.1}
In \cite[Section 1.6]{gelfand05}, the authors showed the solution of linear systems with non-commutative  coefficients could be expressed in terms of quasi-determinants. We start with a linear system with two unknown variables, to demonstrate a way to quasi-Pfaffians. Let $A=(a_{ij})_{i,j=1,2}$ be a $2\times 2$ matrix over a division ring $\mathcal{R}$. For a linear system
\begin{align}\label{linearsys1}
\begin{aligned}
\left\{\begin{array}{l}
a_{11}x_1+a_{12}x_2=b_1,\\
a_{21}x_2+a_{22}x_2=b_2,
\end{array}
\right.
\end{aligned}
\end{align}
where $b_1,b_2\in\mathcal{R}$, the solution could be given by 
\begin{align*}
x_1=-\left|
\begin{array}{ccc}
a_{11}&a_{12}&b_1\\
a_{21}&a_{22}&b_2\\
1&0&\boxed{0}
\end{array}
\right|,\quad x_2=-\left|
\begin{array}{ccc}
a_{11}&a_{12}&b_1\\
a_{21}&a_{22}&b_2\\
0&1&\boxed{0}
\end{array}
\right|.
\end{align*}
In the followings, we first explore the relation between \eqref{linearsys1} and Pfaffians in commutative case. To this end, we assume that $a_{11}=a_{22}=0$, and $a_{12}=-a_{21}$. 
We introduce $A^{ij}$ as a minor deleting $i$-th row and $j$-th column from $A$, a quasi-determinant $|A|_{ij}$ could be expressed as a ratio of determinants. In other words, we have the formula $|A|_{ij}=(-1)^{i+j}\det(A)/\det(A^{ij})$. Thus the solution of \eqref{linearsys1} could be written as
\begin{align*}
x_1=-\frac{\det\left(\begin{array}{ccc}
0&a_{12}&b_1\\
-a_{12}&0&b_2\\
1&0&0
\end{array}
\right)}{\det\left(\begin{array}{cc}
0&a_{12}\\-a_{12}&0
\end{array}
\right)},\quad x_2=-\frac{\det\left(\begin{array}{ccc}
0&a_{12}&b_1\\
-a_{12}&0&b_2\\
0&1&0
\end{array}
\right)}{\det\left(\begin{array}{cc}
0&a_{12}\\-a_{12}&0
\end{array}
\right)}.
\end{align*}
By using the Desnanot-Jacobi identity \cite[Section 2.6.2]{hirota04}
\begin{align*}
\det\left(\begin{array}{cccc}
0&a_{12}&b_1&-1\\
-a_{12}&0&b_2&0\\
-b_1&-b_2&0&0\\
1&0&0&0
\end{array}
\right)\det\left(\begin{array}{cc}
0&a_{12}\\
-a_{12}&0
\end{array}
\right)=\det\left(\begin{array}{ccc}
0&a_{12}&b_1\\
-a_{12}&0&b_2\\
1&0&0
\end{array}
\right)^2,
\end{align*}
we get\footnote{In fact, a general connection between determinants with skew-symmetric structures and Pfaffians was originally given by Cayley \cite{cayley49} . Let's denote $(a_{i,j})_{i,j=2,\cdots,n}$ as a skew-symmetric matrix. For 
even $n$, we have
\begin{align}\label{cayley}
\det\left|\begin{array}{cccc}
a_{xy}&a_{x2}&\cdots&a_{xn}\\
a_{2y}&a_{22}&\cdots&a_{2n}\\
\vdots&\vdots&&\vdots\\
a_{ny}&a_{n2}&\cdots&a_{nn}
\end{array}
\right|=\Pf(x,2,\cdots,n)\Pf(y,2,\cdots,n),
\end{align}
where $\Pf(i,j)=a_{ij},\,\Pf(x,i)=a_{xi},\,\Pf(y,i)=a_{yi}$ and $a_{xi}=-a_{ix}$, $a_{yi}=-a_{iy}$. For odd $n$, we have 
\begin{align*}
\det\left|\begin{array}{cccc}
0&a_{x2}&\cdots&a_{xn}\\
a_{2y}&a_{22}&\cdots&a_{2n}\\
\vdots&\vdots&&\vdots\\
a_{ny}&a_{n2}&\cdots&a_{nn}
\end{array}
\right|=\Pf(x,y,2,\cdots,n)\Pf(2,\cdots,n),
\end{align*}
where we have an extra notation $\Pf(x,y)=0$.}
\begin{align*}
x_1=-\frac{\Pf(1,2,b,c_1)}{\Pf(1,2)},\quad x_2=-\frac{\Pf(1,2,b,c_2)}{\Pf(a_{12})},
\end{align*}
where $\Pf(1,2)=a_{12}$, $\Pf(i,b)=b_i$, $\Pf(i,c_j)=\delta_{i,j}$ and $\Pf(b,c_i)=0$. Thus, applying the expansion formula \eqref{expansion} yields the expression
\begin{align*}
x_1=-\frac{b_2}{a_{12}}, \quad x_2=\frac{b_1}{a_{12}}.
\end{align*}
Now we turn to the non-commutative case. Let's assume $a_{ij}=-a_{ji}^\top$ for all $i,j=1,2$, 
where $a_{11}$ and $a_{22}$ are not necessarily zeros. By using the inverse, we have the solution
\begin{align}\label{sol1}
\left(\begin{array}{c}
x_1\\x_2
\end{array}
\right)=\left(\begin{array}{cc}
a_{11}&a_{12}\\
-a_{12}^\top&a_{22}
\end{array}
\right)^{-1}\left(\begin{array}{c}
b_1\\b_2
\end{array}
\right),
\end{align}
required that the inverse of the coefficient matrix exists. 
Therefore, the solution of the linear system has the quasi-determinant solution
\begin{align*}
x_i=-\left|\begin{array}{cc}
A&b\\
e_i&\boxed{0}\end{array}
\right|,
\end{align*}
where $e_i$ is the $i$-th unit vector, $A$ is the skew-symmetric coefficient matrix and $b$ is the right-hand-side term. It is natural to ask whether we can define a quasi-Pfaffian in terms of quasi-determinant  inspired by the solution. This idea promotes the concept of quasi-Pfaffian in Section \ref{sec2}.
To conclude, we can show that the solution to a linear system governed by a skew-symmetric matrix can be expressed as quasi-Pfaffians.
\begin{theorem}\label{linear}
For a linear system
\begin{align}\label{linearsys}
\sum_{j=1}^{2n}a_{ij}x_j=b_i,\quad i=1,2,\cdots,2n,
\end{align}
where $a_{ij}\in\mathcal{R}$ satisfying $a_{ij}=-a_{ji}^\top$, it admits unique solutions $x_i\in\mathcal{R}$ if and only if $(a_{ij})_{1\leq i,j\leq 2n}$ is invertible. Moreover, the solution could be expressed in terms of quasi-Pfaffian
\begin{align}\label{ci}
x_i=\Pf(1,\cdots,2n,\boxed{c_i,b}),
\end{align}
where $\Pf(i,j)=a_{ij}$, $\Pf(j,b)=b_j$, $\Pf(c_i,b)=0$, and $\Pf(c_i,j)=-\delta_{i,j}$.
\end{theorem}
The commutative case of this linear system was studied by Jacobi in 1827 \cite{jacobi27}, as an analog of ``Cramer's rule'' for the solution of general systems of skew-symmetric linear equations. 
The solution is then given by \cite[Equation 6.1]{knuth96}
\begin{align*}
x_j=\frac{\Pf(1,\cdots,\hat{j},\cdots,2n,b)}{\Pf(1,\cdots,2n)}.
\end{align*}


\begin{thebibliography}{1}

\bibitem{adler99}
M. Adler, E. Horozov and P. van Moerbeke.
\newblock The Pfaff lattice and skew-orthogonal polynomials.
\newblock \emph{Int. Math. Res. Not.}, 11 (1999), 569-588.	
	
\bibitem{adler02}
M. Adler, T. Shiota and P. van Moerbeke.
Pfaff $\tau$-functions.
\emph{Math. Ann.}, 322 (2002), 423-476.	

\bibitem{adler022}
M. Adler and P. van Moerbeke.
Toda versus Pfaff lattice and related polynomials.
\emph{Duke Math J.}, 112 (2002), 1-58.

\bibitem{ariznabarreta14}
G. Ariznabarreta and M. Ma\~nas.
Matrix orthogonal Laurent polynomials on the unit circle and Toda type integrable systems.
\emph{Adv. Math.}, 264 (2014), 396-463.
	
\bibitem{cayley49}
A. Cayley. 
Sur les d\'eterminants gauches.
\emph{J. Reine Angew. Math.}, 38 (1849), 93-96.

\bibitem{chang18}
X. Chang, Y. He, X. Hu and S. Li.	
Partial-skew-orthogonal polynomials and related integrable lattices with Pfaffian tau functions. \emph{Comm. Math. Phys.}, 364 (2018) 1069-1119.

\bibitem{doliwa25}
A. Doliwa.
Non-commutative multiple bi-orthogonal polynomials: formal approach and integrability.
arXiv: 2510.02207.

\bibitem{dress95}	
A. Dress and W. Wenzel.
A simple proof of an identity concerning Pfaffians of skew symmetric matrices.
\emph{Adv. Math.}, 112 (1995), 120-134.

\bibitem{etingof98}
P. Etingof, I. Gelfand and V. Retakh.
Nonabelian integrable systems, quasideterminants, and Marchenko lemma.
\emph{Math. Res. Lett.}, 5 (1998), 1-12.

	
\bibitem{gelfand91}
I. Gelfand and V. Retakh.
Determinants of matrices over noncommutative ring. \emph{Func. Anal. Appl.}, 25 (1991), 91-102.	
	
\bibitem{gelfand92}
I. Gelfand and V. Retakh.
A theory of noncommutative determinants, and characteristic functions of graphs.  \emph{Func. Anal. Appl.}, 26 (1992), 231-246.
	
\bibitem{gelfand05}
I. Gelfand, S. Gelfand, V. Retakh and R. Wilson.
Quasi-determinants.
\emph{Adv. Math.}, 193 (2005), 56-141.

\bibitem{gilson071}
C. Gilson and J. Nimmo. 
On a direct approach to quasideterminant solutions of a noncommutative KP equation.
\emph{J. Phys. A}, 40 (2007), 3839.

\bibitem{gilson072}
C. Gilson, J. Nimmo and Y. Ohta.
Quasideterminant solutions of a non-Abelian Hirota?Miwa equation.
\emph{J. Phys. A}, 40 (2007), 12607.

\bibitem{gilson03}
C. Gilson, X. Hu, W. Ma and H. Tam.
Two integrable differential-difference equations derived from the discrete BKP equation and their related equations.
\emph{Phys. D}, 175 (2003), 177-184.
	
\bibitem{gilson08}
C. Gilson, J. Nimmo and C. Sooman.
On a direct approach to quasideterminant solutions of a noncommutative modified KP equation.
\emph{J. Phys. A}, 41 (2008), 085202.	
	
\bibitem{gilson25}
C. Gilson, S. Li and Y. Shi.
Matrix-valued theta-deformed bi-orthogonal polynomials, non-commutative Toda theory and B\"acklund transformation. \emph{Nonlinearity}, 38 (2025) 065003.	
	
\bibitem{hirota89}
R. Hirota.
Soliton Solutions to the BKP Equations. I. the Pfaffian technique.
\emph{J. Phys. Soc. Jpn}, 58 (1989), 2285-2296.

\bibitem{hirota04}
R. Hirota. \emph{The Direct Method in Soliton Theory}, English transl. by A. Nagai, J. Nimmo, C. Gilson from 1992 Japanese ed., Cambridge University Press, Cambridge, 2004.

\bibitem{hirota00}
R. Hirota, M. Iwao and S. Tsujimoto.
Soliton equations exhibiting Pfaffian solutions.
\emph{Glasgow Math. J.}, 43A (2001), 33-41.

\bibitem{hirota91}
R. Hirota and Y. Ohta.
Hierarchies of Coupled Soliton Equations. I.
\emph{J. Phys. Soc. Jpn.}, 60 (1991), 798-809.

\bibitem{ishikawa95}
M. Ishikawa and M. Wakayama.
Minor summation formula of Pfaffians.
\emph{Linear and Multilinear Algebra}, 39 (1995), 285-305.

\bibitem{ishikawa96}
M. Ishikawa, S. Okada and M. Wakayama.
Applications of minor-summation formula I. Littlewood's formulas.
\emph{J. Algebra}., 183 (1996), 193-216.

\bibitem{jacobi27}
C. Jacobi. \"Uber die \emph{Pfaffsche} Methode, eine gew\"ohnliche line\"are Differential-gleichung
zwischen $2n$ Variabeln durch ein System von $n$ Gleichungen zu integriren. \emph{Journal f\"ur die reine und angewandte Mathematik}., 2 (1827), 347-357. Reprinted in \emph{C. G. J. Jacobi's Gesammelte Werke} 4 (1886), 17-29.

\bibitem{jing14}
N. Jing and J. Zhang.
Quantum Pfaffians and hyper-Pfaffians.
\emph{Adv. Math.}, 265 (2014), 336-361.

\bibitem{kasteleyn61}
P. Kasteleyn. 
The physics of dimers on a lattice, \emph{Physica}, 27 (1961), 1209-1225.


\bibitem{kenyon05}
R. Kenyon and A. Okounkov.
What is ... a dimer?
\emph{Notices of AMS}, 52 (2005), 342-343.

\bibitem{knuth96}
D. Kunth. 
Overlapping pfaffians.
\emph{Electron. J. Comb.}, 3 (1996), R5.

\bibitem{kodama10}
Y. Kodama and V. Pierce.
The Pfaff lattice on symplectic matrices.
\emph{J. Phys. A}, 43 (2010), 055206.

\bibitem{krichever95}
I. Krichever.
 General rational reductions of the Kadomtsev-Petviashvili hierarchy and their symmetries.
 \emph{ Funct. Anal. Appl.}, 29 (1995), 75-80.

\bibitem{krob95}
D. Krob and B. Leclerc.
Minor identities for quasi-determinant and quantum determinants.
\emph{Comm. Math. Phys.}, 169 (1995), 1-23.

\bibitem{lam06}
T. Lam and R. Swan.
Sums of alternating matrices and invertible matrices.
Algebra and Its Applications, Proc. International Conference, Athens, Ohio, March, 2005, D.V. Huynh,
S.K. Jain, S. L´opez-Permouth, eds, \emph{Contemp. Math.}, 419 (2006), 201–211, Amer. Math. Soc.,
Providence, R.I. 
 
 
\bibitem{li20}
S. Li.
Discrete integrable systems and condensation algorithms for Pfaffians, arXiv:2006.06221. 

\bibitem{li242}
S. Li.
Matrix Orthogonal Polynomials, non-abelian Toda lattice and B?cklund transformation.
\emph{ Sci. China Math.}, 67 (2024), 2071-2090.

\bibitem{li22}
S. Li and G. Yu.
Integrable lattice hierarchies behind Cauchy two-matrix model and Bures ensemble.
\emph{Nonlinearity}, 35 (2022), 5109. 

\bibitem{li23}
S. Li, B. Shen, J. Xiang and G. Yu, Multiple skew orthogonal polynomials and 2-component Pfaff lattice hierarchy. \emph{Annales Henri Poincar\'e}, 25 (2024), 3333-3370

\bibitem{li24}
S. Li and G. Yu
Christoffel transformations for (partial-)skew-orthogonal polynomials and applications. \emph{Adv. Math.}, 436 (2024), 109398.

\bibitem{li252}
S. Li, Y. Shi, G. Yu and J. Zhao.
Matrix-valued Cauchy bi-orthogonal polynomials and a novel noncommutative integrable lattice. 
\emph{Stud. Appl. Math.}, 154 (2025) e70040.


\bibitem{ohta92}
Y. Ohta.
\emph{Bilinear theory of soliton}, PhD thesis, Tokyo University, 1992.	

\bibitem{ohta04}
Y. Ohta.
\emph{Special solutions of discrete integrable systems.}
Lecture Notes in Physics, Vol 644, pp 57-83, 2004.
	
\bibitem{perk84}	
J. Perk, H. Capel, G. Quispel and F. Nijhoff.
Finite-temperature correlations for Ising chain in a transverse field.
\emph{Physica}, 123A (1984), 1-49.
	
\bibitem{richardson28}
A. Richardson.
Simultaneous linear equation over a division ring.
\emph{Proc. Lond. Math. Soc.}, s2 (1928), 395-420.
	
\bibitem{stembridge90}
J. Stembridge. 
Non-intersecting paths, pfaffians, and plane partitions.
\emph{Adv. Math.}, 83 (1990), 96-131.	

\bibitem{tanner78}
H. Tanner. 
A theorem relating to pfaffians.
\emph{Messenger Math.}, 8 (1878), 56-59.	

\bibitem{tsujimoto96}
S. Tsujimoto and R. Hirota.
Pfaffian representation of solutions to the discrete BKP hierarchy in bilinear form.
\emph{J. Phys. Soc. Jpn.}, 9 (1996), 2797-2806.	
	
\bibitem{li25}
Z. Yao and S. Li.
Non-intersecting path explanation for block Pfaffians and applications into skew-orthogonal polynomials.
\emph{Adv. Appl. Math.}, 163 (2025), 102803.
	
	
\bibitem{zajaczkowski80}	
W. Zajaczkowski. A theorem relating to Pfaffians.
\emph{Messenger of Mathematics}, 10 (1880), 36-37.
	
\end{thebibliography}
\end{document}